\documentclass[fleqn]{article}

\usepackage[utf8]{inputenc}
\usepackage[T1]{fontenc}    
\usepackage{hyperref}       
\usepackage{url}            
\usepackage{booktabs}       
\usepackage{amsfonts}       
\usepackage{nicefrac}       
\usepackage{microtype}      
\usepackage{graphicx}       
\usepackage{bm}
\usepackage{bbm}
\usepackage{amsthm}
\usepackage{dsfont}
\usepackage{enumitem}
\usepackage{tabularx}
\usepackage[fleqn]{amsmath}
\usepackage[margin=1in]{geometry}
\usepackage{tikz}
\usepackage{listings}
\usepackage{xcolor}
\usepackage{float}
\usepackage[rightcaption]{sidecap}

\usetikzlibrary{shapes.geometric, arrows}

\DeclareMathOperator{\E}{E}

\DeclareMathOperator{\Var}{Var}

\DeclareMathOperator{\argmin}{arg\min}
\DeclareMathOperator{\N}{\mathbb{N}}
\DeclareMathOperator{\R}{\mathbb{R}}
\DeclareMathOperator{\ind}{\mathbbm{1}}

\DeclareMathOperator{\Ber}{Ber}

\DeclareMathOperator{\Dir}{Dir}
\DeclareMathOperator{\diag}{diag}
\DeclareMathOperator{\vecop}{Vec}

\DeclareMathOperator{\asto}{\overset{a.s.}{\longrightarrow}}

\newcommand{\pderiv}[3][]{\frac{\partial^{#1}}{ \partial #2}\left[ #3 \right]}
\newcommand{\pderif}[3][]{\frac{\partial^{#1} #3}{ \partial #2}}
\newcommand{\s}[2] {\sum_{#1}^{#2}}

\newcommand{\fr}[1]{ \left( #1 \right)}
\newcommand{\fb}[1]{ \left[ #1 \right]}
\newcommand{\fc}[1]{ \left\{ #1 \right\}}
\newcommand{\given}[1]{\left. #1 \right|}
\newcommand{\set}[2]{\left\{#1 \left\vert #2 \right. \right\}}
\newcommand{\norm}[1]{\left\lVert #1 \right\rVert}
\newcommand{\abs}[1]{\left| #1 \right|}

\newcounter{todocounter}

\newtheorem{theorem}{Theorem}

\newtheorem{corollary}{Corollary}

\newtheorem{definition}[theorem]{Definition}
\newtheorem{example}[theorem]{Example}

\newtheorem{lemma}{Lemma}

\newtheorem{remark}[theorem]{Remark}

\tikzstyle{att0} = [rectangle, rounded corners, 
minimum width=2cm, 
minimum height=1cm,
text centered, 
draw=black, 
fill=red!30]

\tikzstyle{att1} = [rectangle, rounded corners, 
minimum width=2cm, 
minimum height=1cm,
text centered, 
draw=black, 
fill=blue!30]

\tikzstyle{arrow} = [thick,->,>=stealth]

\title{Attractor-Based Coevolving Dot Product Random Graph Model}
\author{Shiwen Yang, Daniel L. Sussman}

\begin{document}
\allowdisplaybreaks
\maketitle

\begin{abstract}
    \noindent We introduce the attractor-based coevolving dot product random graph model (ABCDPRGM) to analyze time-series network data manifesting polarizing or flocking behavior. Graphs are generated based on latent positions under the random dot product graph regime. We assign group membership to each node. When evolving through time, the latent position of each node will change based on its current position and two attractors, which are defined to be the centers of the latent positions of all of its neighbors who share its group membership or who have different group membership than it. Parameters are assigned to the attractors to quantify the amount of influence that the attractors have on the trajectory of the latent position of each node. We developed estimators for the parameters, demonstrated their consistency, and established convergence rates under specific assumptions. Through the ABCDPRGM, we provided a novel framework for quantifying and understanding the underlying forces influencing the polarizing or flocking behaviors in dynamic network data.
\end{abstract}

\section{Introduction}

 Much research interest in network analysis has gone into static network models, which capture a single snapshot of network interactions. While such models excel at describing any time-invariant data, they have difficulty reflecting evolutions within a network over time, and dynamic network models have been introduced to model such properties \cite{dynamic_network_basics}. This class of models aims to help researchers capture dynamic behaviors, such as the formation and dissolution of nodes and edges over time, in systems like social networks \cite{Dynamic_Social_N_APP_Exp_1} or biological ecosystems \cite{Dynamic_Eco_N_APP_Exp_1}.\\

 This paper will focus on two types of dynamic behaviors: flocking and polarizing. Flocking behavior, observed in phenomena such as birds flying in coordinated formations and fish swimming in schools, involves individuals within a network aligning their actions or states to match those of their neighbors \cite{Flocking_Def}. Polarizing behavior, in contrast, occurs when members of a community increasingly divide into opposing groups, typically leading to increased homogeneity within each group and greater heterogeneity between groups \cite{Polarization_Def}. The study of flocking and polarizing behaviors extends beyond theoretical interest. In biological conservation, for example, detecting the change in mixed species flocking composition highlights the bird trade's threat to the local biodiversity \cite{SongBird_Flocking_Exp}. Meanwhile, researchers have also long been modeling the polarization on social media to study its impact on politics \cite{Political_Polarization_Exp_1} \cite{Political_Polarization_Exp_2} and science \cite{Science_Polarization_Exp} over time. \\

 Latent space models, like the random dot product graph (RDPG) \cite{RDPG_Survey}, have been a popular class of static network models \cite{Latent_Space_Network_Model_Popularity} \cite{Latent_Space_Network_Model_Review}. By representing nodes in the subspace of some Euclidean space\cite{sewell2015latent}, these models capture the hidden structures in the network. Attempts to adapt latent space model to describe dynamic behaviors started by assuming that the latent space is where all the dynamics occur \cite{Naive_Latent_Space_Dynamic_Model_1} \cite{Naive_Latent_Space_Dynamic_Model_2}. This assumption implies that conditioning on the latent positions, the graph structure at time $t$ is independent of the graph structure at time $t - 1$. While such an assumption may be sufficient for specific applications\cite{xu2014dynamic}\cite{olivella2022dynamic}, it fails to capture the most basic assumption for flocking behaviors: each individual makes decisions based on their neighbor’s decision\cite{hartle2021dynamic}. The Coevolving Latent Space Network with Attractors(CLSNA) model \cite{CLSNA} addresses this shortcoming by incorporating attractors at time $t$ that depend on the graph structure at time $t- 1$.\\

 Inspired by CLSNA, we develop a model under the RDPG framework to take advantage of its analytical tractability \cite{RDPG_Survey}. We aim to model the flocking-polarizing behavior in networks. In our K-group model, we assume that each node belongs to one of the K groups, and the movement of each node in the latent space is influenced by their current position, as well as two other attractors determined by the graph structure representing attraction or repulsion from neighbors. \\

 The remainder of the manuscript is organized as follows. In Section \ref{sec: model def}, we introduce the RDPG, present our novel dynamic network model, and discuss the model's behaviors and parameter interpretations. In Section \ref{sec: methodology}, we propose a regression framework for our model and discuss the two steps to estimate the parameters of our model. The first step is to recover the latent positions through adjacency spectral embedding (ASE).
In the second step, with the recovered latent positions, we estimate the parameters that represent potential flocking and polarizing behavior. In Section \ref{sec: results}, we first show that with known latent positions, our estimate is consistent and asymptotically normal. We then show that regression using the ASE estimates of the latent positions can also yield consistent estimates. Finally, we briefly discuss a proposed solution to the non-identifiability problem inherent in using the ASE In Section \ref{sec: real data}, we test our model with a real network data set derived from a competitive online game and show that our method can detect polarizing and flocking behaviors.

\section{The Dynamic Model and Related Definitions} \label{sec: model def}
\begin{table}[!htbp]
    \centering
    \begin{tabular}{c|l}
        \textbf{Notation} & \textbf{Definition} \\
        \hline
        $\mathbb{H}$ & $\R^+$ \\
        \hline
        $I_p$ & The $(p \times p)$ identity matrix.\\
        \hline
        $\mathbf{e}_p$ & The $p^{th}$ standard basis of $\R^q$ for some $q \geq p$. The exact value of $q$ will depend on the context.\\
        \hline
        $\mathbf{1}_p$ & The length-$p$ all-one vector.\\
        \hline
        $\Delta^p$ & $\set{x \in \mathbb{H}^p}{x^T\mathbf{1}_p\leq 1}$\\
        \hline
        $\mathbf{0}_{p \times q}$ & The dimension-$(p\times q)$ all-zero matrix.\\
        \hline
        $\mathbbm{1}_{\text{condition}}$ & The indicator function for the referenced condition in the subscript.\\
        \hline
        $\otimes$ & The Kronecker product\\
        \hline
        $\vecop$ & The vectorization operator -- the canonical projection from $\R^{m \times n}$ to $\R^{mn}$.\\
        \hline
        $\abs{S}$ & For a set $S$, $\abs{S}$ denotes the cardinality of $S$.\\
        \hline
        $O_p$ & The space of $p \times p$ real-valued orthogonal matrices\\
        \hline
    \end{tabular}
    \caption{Table of Notations}
    \label{table: notations}
\end{table}

 Let $p\geq 1$ be an integer denoting the dimension of the latent positions. Let $Z_1,\dotsc, Z_n$ be $\R^p$ random vectors such that $\forall i, j$, $Z_i^TZ_j \in [0,1]$ almost surely. Collect $Z_1, \dotsc, Z_n$ in the rows of an $\R^{n \times p}$ random matrix $Z$. We write $Y\sim \mathrm{RDPG}(Z)$ if $Y$ is a symmetric $n \times n$ random matrix with the following property\cite{RDPG_Survey}, and also note that conditioning on the latent positions, the entries of $Y$ are independent Bernoulli random variables: 
\begin{equation*}
    P\fr{Y | Z} = \prod_{i \leq j \leq n}\fr{ZZ^T}^{Y_{ij}}\fr{1 - ZZ^T}^{1 - Y_{ij}}.
\end{equation*}

 \noindent Let $\fc{Y_t}^{T}_{t = 0}$, be a sequence of RDPG with common set of vertices $V$, and let $\fc{Z_{t}}_{t = 0}^T$ be the corresponding latent positions\footnote{Each $Z_{t}$ is a $n \times p$ matrix. When we want to refer to some component of $Z_t$, the time index, $t$, will always be the second index, e.g. $Z_{ij, t}$ is the $ij^{th}$ component of $Z_t$}. In addition to the latent positions, we assign a group membership to each node with a function $\pi: V \to \mathcal{C}$ that maps vertices from the set of vertices $V$ to the set of group labels $\mathcal{C}$. For each node $i$, define $\tau_{w}(i) = \pi^{-1}(\pi(i)) - \fc{i} \subset [n]$, and $\tau_b(i) = \pi^{-1}\fr{\mathcal{C} - \fc{\pi(i)}} \subset [n]$. These are the sets of groupmates/non-groupmates of node $i$. We then define the intra-group attractor of node $i$, which is the average of the latent positions of all neighbors of $i$ with the same group membership:
\begin{alignat}{2}
    A^w_{i}(Z_{t}, Y_{t}) = \begin{cases}
        0 \quad &\text{if } k := \sum\limits_{j \in \tau_w(i)}Y_{ij} = 0 \\
        \frac{1}{k}\sum\limits_{j \in \tau_w(i)} Y_{ij}Z_{j*, t} \quad &\text{otherwise}
    \end{cases}.
\end{alignat}
\noindent The inter-group center $A^b_{i}(Z_{t}, Y_{t})$ is defined similarly but uses $\tau_b(i)$ instead. In addition, we shall add a superscript star, e.g. $A^{w*}_i(Z_t, Y_t)$ to indicate the inclusion of the $(p+1)^{th}$ dimension.\footnote{$\forall i \in V,\ Z_{i*,t}$ is a p-dimensional vector such that $\sum_{j= 1}^{p}Z_{ij, t} \leq 1$. One more dimension is added to $Z_{i*, t}$ to make $Z^*_{i*, t}$ so that $\sum_{j= 1}^{p+1} Z^*_{ij, t} = 1$. $A^{w*}_{i}(Z_t, Y_t)$ is defined similarly.} We shall omit the arguments of these two functions and add time $t$ to the subscript for brevity, i.e. we shall use $A^{w*}_{i, t}, A^{b*}_{i, t}$ instead.

 Using these building blocks, we define how the latent position changes over time. At time $0$, all latent positions are independent Dirichlet random variables with parameters that are i.i.d. random variables distributed on $\mathbb{H}^{p + 1}$\cite{dirichlet_definition}. In other word, let $F$ be a distribution such that $\text{supp}\fr{F} \subset \mathbb{H}^{p+1}$, then for $i = 1,\dotsc, n$, $Z^*_{i, 0} \sim \text{Dir}(\alpha_{i, 0})$ where $\alpha_{i, 0} \overset{i.i.d.}{\sim} F$. At time $t + 1$: $Z^*_{i, t+1} \sim \text{Dir}(\alpha_{i, t+1})$, where $\alpha_{i, t+1} = \exp\fc{\beta_1 Z^*_{i, t} + \beta_2 A^{w*}_{i,t} + \beta_3 A^{b*}_{i,t} + \beta_4}$, and $\beta_1, \dotsc, \beta_4 \in \R$. Finally, the RDPG at time $t$ is given by: $Y_{ij, t}| Z_{i, t}, Z_{j, t} \sim \Ber\fr{Z_{i, t}^TZ_{j, t}}.$

\begin{figure}
\begin{center}
    \begin{tikzpicture}[node distance=2cm]
        \node (Zt0) [att0] {$Z_{t-1}$};
        \node (Yt0) [att0, below of = Zt0, xshift = -2cm] {$Y_{t-1}$};
        \node (At0) [att0, above of = Zt0, xshift = -2cm] {$A^w_{t}, A^b_{t}$};
        \node (Xt0) [att0, above of = Zt0, xshift = 2cm] {$X_{t-1}$};
        \node (Zt1) [att1, right of = Zt0, xshift = 5cm] {$Z_{t}$};
        \node (Yt1) [att1, below of = Zt1, xshift = -2cm] {$Y_{t}$};
        \draw [arrow](Zt0) -- node[anchor = west] {$\text{Ber}(Z_{t-1}Z_{t-1}^T)$} (Yt0);
        \draw [arrow](Zt0) -- (At0);
        \draw [arrow](Zt0) -- (Xt0);
        \draw [arrow](Yt0) -- (At0);
        \draw [arrow](At0) -- (Xt0);
        \draw [arrow](Xt0) |- node[anchor = west, yshift = 0.5cm] {$Z^*_{t} \sim \Dir\fr{\exp\fr{X_{t}B}}$} (Zt1);
        \draw [arrow](Zt1) -- node[anchor = west] {$\text{Ber}(Z_{t}Z_{t}^T)$} (Yt1);
        \end{tikzpicture}
            \caption{This is a graph representation of our model. The annotated lines indicate randomness in the relationships whereas the lack thereof represent deterministic relationships. Although presented later in Equation \ref{eq: design matrix}, $X_t,\ B$ are defined such that $\exp{X_{i*,t}^TB} = \alpha_{i, t+1}$}
            \label{fig:FC1}
\end{center}
\end{figure}
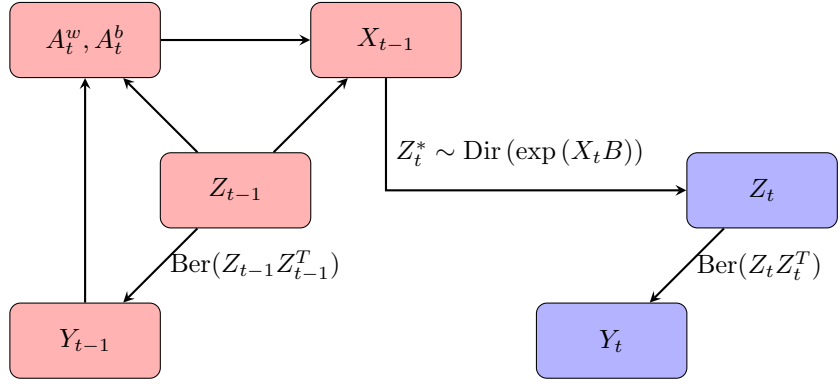

 The attractors are introduced to model the expected polarizing/flocking behavior induced by the graph structure. They are defined for each node to represent the influence exerted by different parties on each node through its connections. For each time step, each node will move according to how much it is influenced by the different parties, which is quantified by a parameter, $\beta = \begin{bmatrix}
    \beta_1 & \beta_2 & \beta_3 & \beta_4
\end{bmatrix}$. For an arbitrary node $i$:
\begin{enumerate}
    \item $\beta_1$ quantifies the influence of the latent position of node $i$ at time $t-1$ to its latent position at time $t$.
    \item $\beta_2$ quantifies the influence from all neighbors of $i$ who are in the same party as $i$ to node $i$.
    \item $\beta_3$ quantifies the influence from all neighbors of $i$ who are in a different party than $i$ to node $i$.
    \item $\beta_4$ is a nuisance parameter characterizing the change in variance from one-time point to the next.
\end{enumerate}

 Since the flocking/polarizing behavior happens at the groups level, the size of $\beta_2$ determines the rate of flocking, i.e. how fast each group is contracting, and the sign of $\beta_3$ will determine the type of behavior that the model will display. Large value of $\beta_2$ corresponds to a fast rate of flocking within each group. $\beta_3 > 0$ means that each node will be attracted to the latent position of all its neigbors with different group membership, i.e., all latent positions will get closer, and the model will display flocking behavior. In contrast, when $\beta_3 <0$, every node will be repelled from the latent positions of all its neighbors with different group membership, so the latent positions of nodes with different group membership will grow further apart, thus resulting in a polarized model. Figure \ref{fig: Polarization Exp} shows an example of the evolution in the latent space for a polarizing model.

 Define $\text{softmax}_\lambda(x) = \frac{\exp\fc{\lambda x}}{\sum\limits_{i = 1}^{p}\exp\fc{\lambda x_i}}$ to be the softmax function with parameter $\lambda$. Recall that our model is inherently a Dirichlet GLM with a log-link\cite{mccullagh1989generalized}, i.e. $Z_{i, t+1}^* \sim \Dir\fr{\exp\fc{\beta_1 Z^*_{i, t} + \beta_2 A^{w*}_{i,t} + \beta_3 A^{b*}_{i,t} + \beta_4}}$. The link is a necessary component of our model because the support of the Dirichlet distribution is $\mathbb{H}^{p+1}$, but components of $\beta$ can be negative. However, because of the log-link, there is no $\beta$ such that for all $Z^*_{i,t}$:
\begin{align*}
    \E\fr{Z_{i, t+1}^*|Z^*_{t}} = \text{softmax}_1\fr{\beta_1 Z^*_{i, t} + \beta_2 A^{w*}_{i,t} + \beta_3 A^{b*}_{i,t} + \beta_4} = Z^*_{i, t}
\end{align*}
\noindent While this is an inevitable consequence due to the necessity of the link function, all other aspects of our model behaves intuitively, e.g. predictor with bigger parameter will exert bigger influence. It is also worth noting that the conditional mean of $Z_{i, t+1}^* \sim \Dir\fr{\exp\fc{Z^*_{i, t}}}$, i.e. when $\begin{bmatrix}
    \beta_1 & \beta_2 & \beta_3 & \beta_4
\end{bmatrix} = \begin{bmatrix}
    1 & 0 & 0 & 0
\end{bmatrix}$, which is given by the $\text{softmax}_1$ function, always converges to the barycenter of the standard $p-$simplex when applied iteratively because of the Banach fixed point theorem \cite{banach1922operations}\cite{gao2017properties}. This indicates that under this simplistic setup where the influence of the group centers are absent, the latent positions are expected to be around the barycenter as time progresses regardless of initialization. 

 In practice, we will only observe a time-series network data in the form of a sequence of adjacency matrices. Our goal is to estimate $\beta_1, \beta_2, \beta_3, \beta_4$.  See Table \ref{table: definitions} for a summary of all the definitions, and see Figure \ref{fig:FC1} for the relationship among all the variables listed above. Note, we often omit the time index when it is unecessary.

\begin{figure}[H]
    \centering
    \includegraphics[width=0.8\textwidth]{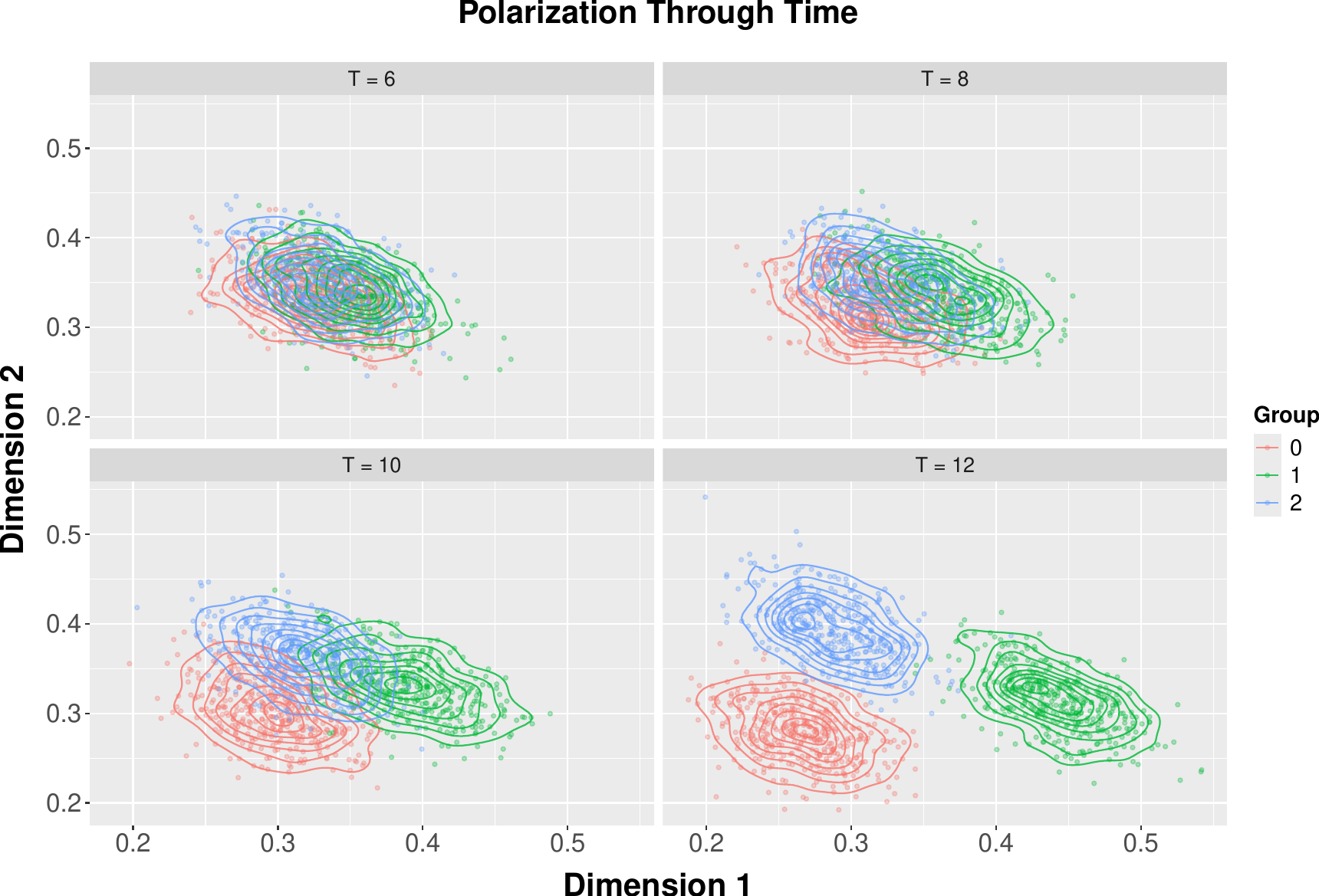}
    \caption{This is an example of latent position polarizing over time. For this simulation, we used $\beta = [1,1,-4,5]$, initialized at $\Dir(\begin{bmatrix}
        1&1&1
    \end{bmatrix})$. We can see that the latent positions are well-mixed for a while in the beginning, then the groups start to distance themselves from other groups. By $t = 12$, the latent positions are almost completely separated by group.}
    \label{fig: Polarization Exp}
\end{figure}

\begin{table}[H]
    \centering
    \begin{tabular}{l|l}
        \textbf{Variable} & \textbf{Definition} \\
        \hline
        $F$ & A distribution on $\mathbb{H}^{p + 1}$\\
        \hline
        $\alpha_{i, 0} \sim F$ for $i = 1,..., n$ & The parameters of latent positions at $t = 0$\\
        \hline
        $Z^*_{i,t}$ & Latent position of vertex $i$ at time $t$ with an extra dimension \\
        \hline
        $Z_{i, t}$ & the first $p$ dimensions of $Z^*_{i,t}$ \\
        \hline
        $Z_t = [Z_{1, t} \dots Z_{n, t}]^T$     & The $(n \times p)$ matrix of latent positions of the entire graph at time $t$\\
        \hline
        $P_t = Z_tZ^T_t$      & The random $(n \times n)$ parameter matrix of edge probabilities\\
        \hline
        $Y_t$ & The adjacency matrix of $G_t$ given by $Y_{ij, t}| P_{ij, t} \sim \Ber\fr{P_{ij, t}\ind\fc{i \neq j}}$\\
        \hline
        $\mathcal{C} \subset \N$ & A finite set of group labels. $|\mathcal{C}| \geq 2$\\
        \hline
        $\pi: V \to \mathcal{C}$ & A function that returns the group label of each vertex.\\
        \hline
        $\tau_w: V \to 2^V$ & Returns the set of all groupmates of each vertex, $\tau_w(i) = \pi^{-1}(\pi(i)) - \fc{i}$\\
        \hline
        $\tau_b: V \to 2^V$ & Returns the set of all non-groupmates for each vertex: $\tau_b(i) = \pi^{-1}(\mathcal{C} - \fc{\pi(i)})$\\
        \hline
         $A^w_{i}(Z_{t}, Y_{t}),\ A^b_{i}(Z_{t}, Y_{t})$ & The within/between group attractors of vertex $i$. Also appear as $A^w_{i,t},\ A^b_{i,t}$\\
         \hline
         $\alpha_{i, t+1}$ & $\alpha_{i, t+1} = \exp\fc{\beta_1 Z^*_{i, t} + \beta_2 A^{w*}_{i,t} + \beta_3 A^{b*}_{i,t} + \beta_4}$\\
         \hline
         
    \end{tabular}
    \caption{Table of Definitions}
    \label{table: definitions}
\end{table}

\section{Methodology} \label{sec: methodology}
\subsection{Overview}
 Recall that $\beta = 
\begin{bmatrix}
    \beta_1 & \beta_2 & \beta_3 & \beta_4
\end{bmatrix}^T$ qauntifies the linear relationship between $Z_t$ and $\log\fr{\alpha_{t+1}}$. After observing a time series of adjacency matrices $\{Y_t\}_{t = 1}^{T}$, we are interested in estimating $\beta$. The estimation is done in two steps. First we solve a minimization problem to estimate the latent position at time $t$ and $t+1$ using the observed adjacency matrices, and then we fit the estimated latent positions to a Dirichlet GLM to obtain our desired estimate of $\beta$. In this paper, we consider the case with two time points, $t = 0, 1$. When there are more time points, we can estimate $\beta_t$ for each $t$ by iteratively fitting our model for every two consecutive time points. Doing so not only gives us an estimate for $\beta$, but also naturally detects abrupt changes in $\beta$ if it changes with respect to time. When the context is clear, we will usually omit the time index for convinience.



In the following sections, we first set up the Dirichlet GLM assuming the latent positions are known. Using existing GLM theory\cite{GLM_Theory}, we prove sufficient conditions for consistency and asymptotic normality for our estimated $\beta$. Then, we show that our estimate of $\beta$ is consistent when our initial estimate of the latent position, which is always off by an orthogonal transformation, is aligned to the true latent position by an oracle. Incrementally weakening the problem this way is necessary because the Dirichlet distribution is not invariant under orthogonal transformation, and aligning our estimated latent position is inherentally nontrivial. Finally, we tackle the problem without oracle information. We prove sufficient conditions for the consistency of our estimate, and provide evidence via simulation that our estimates can be efficient as well.

\subsection{Regression Framework}\label{sec: reg framework}
 The core of the dynamics in our model is a Dirichlet GLM with a log link. If we observe the latent positions alongside the network, then we can form a design matrix, $X_0$, as a function of $Z_0$, and fit a Dirichlet GLM with $X_0$ being the design matrix and $Z_1$ being the response. Below, we construct a design matrix that facilitates applications of existing GLM theory. Let
$\beta_{-4} = 
\begin{bmatrix}
    \beta_1 & \beta_2 & \beta_3
\end{bmatrix}^T$. By the definition of $\alpha_{i, t+1}$, we have:\\

\begin{align*}
        \log\fr{\alpha_{i, t+1}} &= \beta_1 Z^*_{i, t} + \beta_2 A^{w*}_{i,t} + \beta_3 A^{b*}_{i,t} + \beta_4\\
        &= 
    \begin{bmatrix}
        \beta_{-4}^T \otimes I_p & \beta_4\mathbbm{1}_p\\
        \beta_{-4}^T \otimes (-\mathbbm{1}_p) &\mathbbm{1}_4^T B 
    \end{bmatrix}
    \begin{bmatrix}
            Z^T_{i,t} & A^{wT}_{i,t} & A^{bT}_{i,t} & 1
        \end{bmatrix}.
\end{align*}

Define the following terms :
\begin{align} \label{eq: design matrix}
    X_{i, t} = \begin{bmatrix}
        Z_{i,t} \\
        A^{w}_{i,t} \\
        A^{b}_{i,t} \\
        1
    \end{bmatrix}, \quad X_{t} = 
    \begin{bmatrix}
        X_{1, t}^T \\
        X_{2, t}^T \\
        \dots \\
        X_{n, t}^T
    \end{bmatrix}=
    \begin{bmatrix}
        Z_t & A^w_t & A^b_t &\textbf{1}_n \\
    \end{bmatrix}, \quad B = 
    \begin{bmatrix}
        \beta_{-4}^T \otimes I_p & \beta_4\mathbbm{1}_p\\
        \beta_{-4}^T \otimes (-\mathbbm{1}_p) &\mathbbm{1}_4^T \beta 
    \end{bmatrix}^T.
\end{align}
Hence, $\alpha_{t+1} = \exp\fc{X_t B}$ where the exponential is taken element-wise. Note that the $i^{th}$ row of $\alpha_{t+1}$ is the parameter for $Z^{*T}_{i*, t+1}$,  the $i^{th}$ row of $Z^*_{t+1}$. \\

 Under this setting, $X_t\in \R^{(3p+1)\times n}$ is our design matrix, and $B \in \R^{(3p + 1) \times (p + 1)}$ is the parameter matrix of interest, and our model is $Z^*_{i, t+1} \sim \Dir\fr{ \exp\fc{X_{i,t}^T B}}$ for $i = 1, \dots, n$. We maximize the log-likelihood function $\ell \fr{B | Z_{t+1}, X_t}$ using the Fisher's Scoring Algorithm\cite{longford1987fast}. We make the following transformation to our model so that the score function, and the Fisher's information can be expressed using one vector and one matrix correspondingly:
\begin{align*}
\vecop(\log\fr{\alpha_{t+1}}) &= \vecop\fr{X_tB} =  (X_t \otimes I_{p+1})\vecop(B).
\end{align*}
\noindent Although $B_v := \vecop(B)$ is a vector in $\R^{(3p + 1)(p + 1)}$, it is really $\beta \in \R^4$ embeded in $\R^{(3p + 1)(p + 1)}$ through a linear transformation: $B_v = C\beta$, for some fixed matrix $C \in \R^{(3p + 1)(p + 1)\times 4}$. Our ultimate goal is to estimate $\beta$, which can be done via the following steps:
\begin{enumerate} \label{eq:get_beta_from_B}
    \item Obtain $\widehat{B}_v$, an estimate of $B_v$, by fitting the Dirichlet GLM via likelihood maximization with $X_t \otimes I_{p+1}$ as the design matrix.
    \item Get $\widehat{\beta}$, the estimate of $\beta$, by projecting $\widehat{B}_v$ to the column space of $C$, i.e. $\widehat{\beta} = \fr{C^T C}^{-1}C^T\widehat{B}_v$.
\end{enumerate}
 In the following sections, we will derive the consistency and asymptotic normality of $\widehat{B}_v$ by showing that the design matrix has the desired properties to apply existing GLM theory. Since $\widehat{\beta}$ is a linear function of $\widehat{B}_v$, its consistency and asymptotic normality follows those of $\widehat{B}_v$.\\

 The log-likelihood function $\ell(B)$, score function $s_n(B)$, and Fisher's information matrix $F_n(B)$ for our problem are given below. For more details about the Dirichlet GLM, please see Appendix \ref{apd: Dir_GLM}. We shall omit the time subscript from now on. We assume that the design matrix, $X$, is from time $t = 0$, and the response matrix, $Z$, is from time $t = 1$.
Also, in what follows $\log$, the gamma function, $\Gamma$, and the first and second derivatives of the log-gamma function, $\psi$ and $\psi^{(1)}$ respectively, are all applied element-wise.
\begin{align*}
    \ell\fr{B| Z^*} &= \s{i = 1}{n} \alpha_i^T \log(Z_{i*}) - \fb{ \textbf{1}_{p+1}^T \log\fr{\Gamma\fr{\alpha_i}} - \log\fr{ \Gamma\fr{ \textbf{1}_{p+1}^T \alpha_i}}} - \textbf{1}_{p+1}^T \log\fr{Z^*_{i*}}\\
    \pderif{B_v}{\ell\fr{B| Z^*}}= s_n(B) &= \s{i = 1}{n} \fr{X_{i*} \otimes I_{p+1}} \diag \fr{\alpha_i}\fr{\log(Z^*_{i*}) - \mu_i(\alpha_i)}\\
    \pderif[2]{B_v\partial B_v^T}{\ell\fr{\beta| Z^*}} = F_n(B) &= \s{i = 1}{n} \fr{X_{i*} \otimes I_{p+1}} \diag\fr{\alpha_i}\Sigma_i(\alpha_i)\diag\fr{\alpha_i} \fr{X_{i*}^T \otimes I_{p+1}}
\end{align*}
where
\begin{align*}
    \alpha_i &= \exp\fc{\fr{X_{i*}^T \otimes I_{p+1}}B_v} = \exp \fc{X_{i*}^T B}\\
    \mu_i(\alpha_i) &= E(\log(Z^{*}_{i*} | X_{i*})=\psi\fr{\alpha_i} - \psi \fr{\textbf{1}_{p+1}^T\alpha_i }\\
    \Sigma_i(\alpha_i) &= \mathrm{Var}(\log(Z^{*}_{i*} | X_{i*}) =\diag\fr{\psi^{(1)}\fr{\alpha_i}} - \psi^{(1)} \fr{\textbf{1}_{p+1}^T\alpha_i }.
\end{align*}

 Standard GLM theory requires the design matrix to have full rank as well as independent rows. A close examination of $A^w,\ A^b$ reveals that the rows of $X$ are all dependent on each other through the adjacency matrix $Y$:
\begin{align*}
    A_i^w(Z, Y) = \frac{1}{|S_w(i)|}\sum_{j \in S_w(i)} Z_{j} = \frac{\sum_{j \in \tau_w(i)} Y_{ij}Z_{j}}{\sum_{j \in \tau_w(i)} Y_{ij}}.
\end{align*}
Later we will show that conditioning on the latent positions, which are i.i.d., the rows of our design matrix are independent asymptotically. This allows us to prove almost sure consistency and asymptotic normality for our estimator.

\subsection{Estimating the Latent Positions}
 The procedure described above requires knowing the latent positions, $Z_0, Z_1$, but, in reality we rarely have access to the true latent positions. To overcome this, we estimate $Z_0, Z_1$ using the ASE of the adjacency matrices $Y_0, Y_1$. Call the estimates $\widehat{Z}_0, \widehat{Z}_1$. One issue of using ASE to construct the design and response matrix is the inherent non-identifiability problem from RDPG. In Theorem \ref{theorem: oracle align consistency}, we show that our estimate would be consistent if we use an ASE-estimated latent position that is aligned to the true latent position. We propose two methods to address the identifiability problem in practice.\\

 The idea is that the true latent positions are always all inside of $\Delta^p $ as defined in Table \ref{table: notations}. The estimated latent position with the correct alignment should thus have as few points outside of $\Delta^p$ as possible. So we minimize the out-of-simplex penalty, as defined below, for our ASE estimate to get a more reasonable estimate of the latent positions.

\begin{definition} \label{definition: out of simplex penalty}
    For $Z \in \R^{n \times p}$, define it's out-of-simplex penalty to be:
    \begin{align*}
    L_\mu(Z) = \sum_{i = 1}^{n}\sum_{j = 1}^{p}\mathrm{softplus}_\mu(-Z_{ij}) + \sum_{i = 1}^{n}\mathrm{softplus}_\mu\fr{\sum_{j = 1}^pZ_{ij} - 1},
\end{align*}
where $\mathrm{softplus}_\mu(x) = \frac{1}{\mu} \log\left(1 + e^{\mu x}\right)$.
\end{definition}

Recall that $\mathrm{softplus}_\infty(x) = \text{ReLU}(x) = x\mathbbm{1}_{\fc{x > 0}}$. The first sum penalizes the matrix $Z$ for each negative component that it has, and the second sum penalizes $Z$ for each row whose sum is greater than $1$. Since $L$ is symmetric under permutations, the aforementioned non-identifiability issue persists, but only with respect to permutations now. This is sufficient for our purposes because while the set of Dirichlet distributions is invariant under permutations, and $\beta$ does not change when permutations are applied to data.

We propose two methods to use this loss function to achieve embeddings that primarily lie in the simplex.

Regular ASE is equivalent to minimizing the reconstruction error of the adjacency matrix $A$ in the Frobenius-norm sense. Compared to the regular ASE, the following minimization problem removes the diagonal terms by introducing the matrix $M = \abs{I_n - \mathbf{1}_n\mathbf{1}_n^T}$\cite{GD_ASE},
\begin{align*}
    \argmin_{Z \in \R^{n \times p}} \norm{M \circ (A - Z Z^T)}_F^2.
\end{align*}

\noindent Our Gradient-base Adjacency Embedding with Peinalization(GAEP), further modifies this approach. GAEP favors estimated latent positions that are inside of $\Delta_p$ because of the added penalty function. As in \cite{GD_ASE}, the GAEP can be obtained via gradient descent. 
\begin{definition} \label{definition: GAEP}
    For an adjacency matrix $A \in \fc{0, 1}^{n \times n}$, $\lambda > 0$, define its adaptive adjacency spectral embedding to be:
    \begin{align*}
    \argmin_{Z \in \R^{n \times p}} \norm{M \circ (A - Z Z^T)}_F^2 + \lambda L_\mu(Z).
\end{align*}
\end{definition}


Alternatively, we can first compute the ASE, $\widehat{Z}$, as normal, and then find an orthogonal matrix $W$ such that $\widehat{Z}W$ minimizes the ``out-of-simplex'' penalty.

\begin{definition}[Simplicial Adjacency Embedding (SAE)]
    Let $A \in \fc{0, 1}^{n \times n}$ be a adjacency matrix, $\widehat{Z}$ be its $p-$dimensional ASE, then its simplicial adjacency spectral embedding is given by $\widehat{Z}\widehat{W}$, where $\widehat{W} = \argmin_{W \in O_p} L_\mu\fr{\widehat{Z}W}$.
\end{definition}

 Since $O_p$, the space of $p \times p$ orthogonal matrices is a Riemannian manifold, we can use Riemannian gradient descent to find $\widehat{W}$ \cite{Op_RGD}. For more details about the Riemannian gradient descent, please see Section \ref{apd: RGD}.\

GAEP and SAE offers two distinct ways to estimate the latent positions of a graph, with the constraint that the latent position should be inside of the simplex as much as possible. Let $A \in \fc{0,1}^{n \times n}$ be an adjacency matrix. As a minimization problem, the vanilla ASE (in $p$ dimension) minimizes the reconstruction error of $A$ under a rank $p$ constraint, and the solution comes in form of equivalence classes where $Z\in \R^{n \times p}$ is equivalent to $\Tilde{Z}$ if $\exists W \in O_p$ such that $Z = \Tilde{Z}W$.\\

 While SAE is the set of latent positions inside the equivalence class induced by ASE (thus maintaining the optimal reconstruction error) that minimizes the out-of-simplex penalty, GAEP sacrifices the optimal reconstruction error that ASE offers so that the out-of-simplex penalty can be lowered even more. When the true latent positions are completely within the standard simplex, then SAE works great. It is very fast computationally, and under certain conditions we have evidence to believe that it is a consistent estimate of the true latent position as well. However, when the true latent positions are not limited inside of the standard simplex, i.e. when the model is mispecified, a large portion of the SAE may be outside of the simplex. These estimate cannot be used as data for the subsequent regression analysis. GAEP is helpful in this case because it offers a way to balance the reconstruction error and the out-of-simplex penalty. 

\begin{example}
 The plot below shows the difference between each method of estimating the latent positions. In this example, there are 3 groups, and $p = 2$. Within each group, the latent positions are i.i.d. Dirichlet random variables. The parameters are $(1, 1, 10), (1, 10, 1), (10, 1, 1)$ for group $0, 1, 2$ respectively. The true latent positions are plotted in the bottom right corner. ASE without alignment correctly estimates the overall shape of the latent positions, but it is off by an orthogonal transformation, as we discussed previously.\\

 In the oracle case for ASE, we have the true latent position to help us address the non-identifiability issue by solving the orthogonal Procrustes problem: $\min\limits_{W \in O_p}\norm{Z - \widehat{Z}W}_{F}$. It is visualized in the top right corner. It has the same shape and orientation as the true latent position, but it also contains some noise, as indicated by the fuzzy edges and corners.\\

 Finally, on the bottom left corner is the estimation from RGD. In this example, its performace is close to that of the aligned ASE, except it is off by a permutation. This does not create any problem as long as our estimated $Z_t, Z_{t+1}$ are both off by the same permutation and this can be done in practice by initializing $\widehat{Z}_{t+1}$ at $\widehat{Z}_t$.
\begin{figure}
    \centering
    \includegraphics[width=0.8\textwidth]{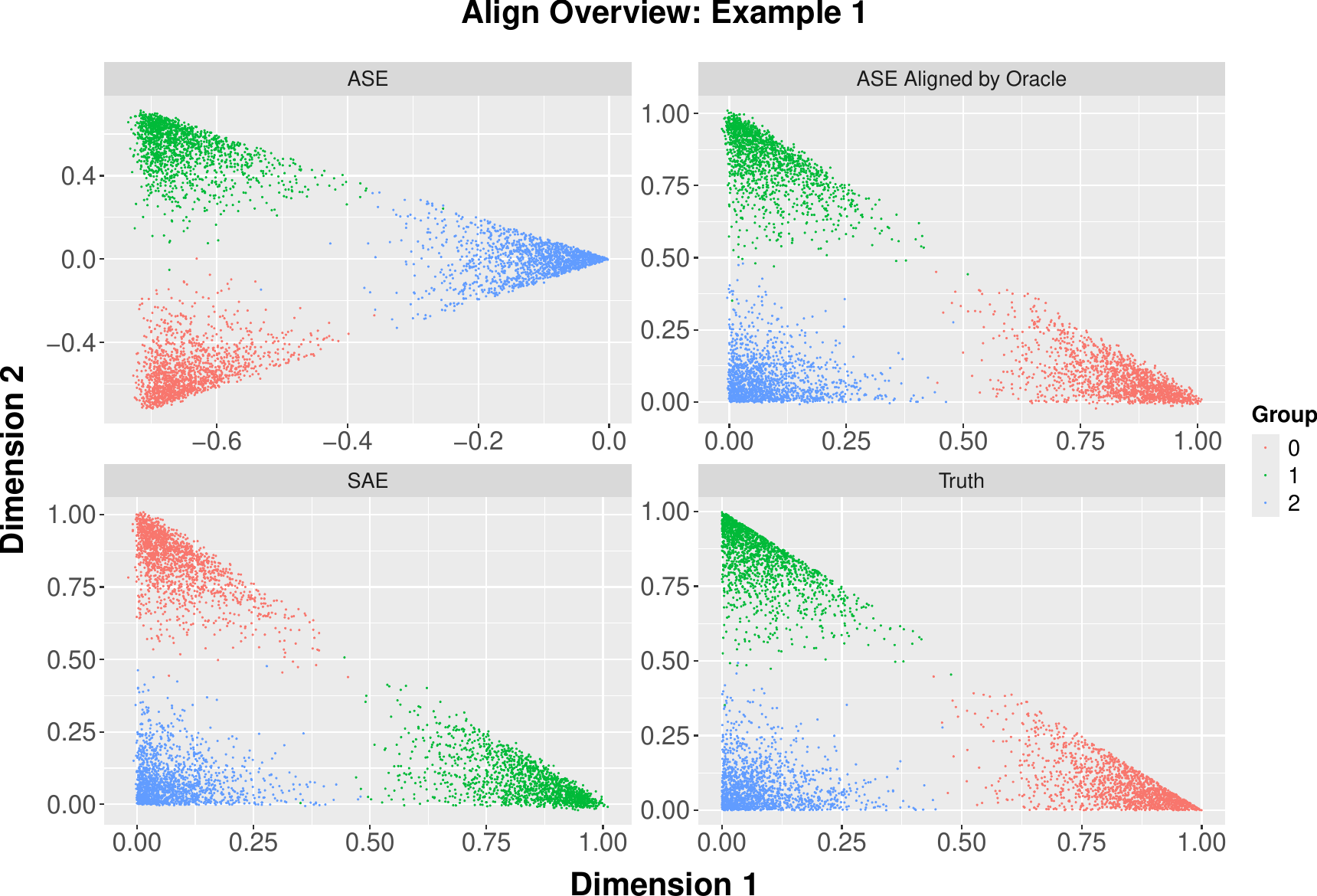}
    \caption{Comparison of Latent Position Estimation Methods}
    \label{fig: Align Methods}
\end{figure}

\end{example}

\section{Main Results} \label{sec: results}
 There are two major parts of our results. First we show that with oracle latent positions, $\widehat{B}$ is consistent and asymptotically normal. Then we show that if we use ASE that is aligned to the true latent position, then our estimate, $\Tilde{B}$, is still consistent. \\

 Before proceeding, we shall define the following terms related to $A^w$ (we omit the $A^b$-counterparts due to their similarity):
\begin{alignat*}{3}
    &N_i 
    = \sum_{j \in \tau_w(i)} Z_{j}Y_{ij}, \quad 
    &&N^*_i
    = \E\fr{\left.\sum_{j \in \tau_w(i)} Z_{j}Y_{ij} \right| Z }, \quad 
    &&\widehat{N}_i 
    = \E\fr{\left.\sum_{j \in \tau_w(i)} Z_{j}Y_{ij} \right| Z_i }\\
    &D_i 
    = \sum_{j \in \tau_w(i)} Y_{ij}, \quad 
    &&D^*_i 
    = \E\fr{\left.\sum_{j \in \tau_w(i)} Y_{ij} \right| Z }, \quad 
    &&\widehat{D}_i 
    = \E\fr{\left.\sum_{j \in \tau_w(i)} Y_{ij} \right| Z_i }\\
    &A^w_i 
    = N_i D^{-1}_i, \quad
    &&A^{w*}_i 
    = N^*_i D^{*-1}_i, \quad
    &&\widehat{A}^w_i 
    = \widehat{N}_i\widehat{D}^{-1}_i.
\end{alignat*}
Recall that $\tau_w$ is defined in Table \ref{table: definitions}.

\subsection{Estimate is consistent with oracle latent position} \label{def: X_star}
 We first show that with oracle latent positions, our estimate, $\widehat{B}$ is consistent and asymptotically normal. This result is only possible because the dependency among the rows of $X$, which originates from the attractors, vanishes asymptotically (assuming row$Z_{0}$ are i.i.d.). Using the definitions above, we define $\widehat{X}_{i,t}$, the i.i.d. version of $X_{i,t}$:
\begin{align*} 
    \text{For $i = 1,..., n$: }\widehat{X}_{i, t} = \begin{bmatrix}
        Z_{i,t} \\
        \widehat{A}^{w}_{i,t} \\
        \widehat{A}^{b}_{i,t} \\
        1
    \end{bmatrix}, \quad \widehat{X}_{t} = 
    \begin{bmatrix}
        \widehat{X}_{1, t}^T \\
        \widehat{X}_{2, t}^T \\
        \dots \\
        \widehat{X}_{n, t}^T
    \end{bmatrix}=
    \begin{bmatrix}
        Z_t & \widehat{A}^w_t & \widehat{A}^b_t &\textbf{1}_n \\
    \end{bmatrix}.
\end{align*}

\noindent For Theorem \ref{theorem: suff cond for consistency}, we will first prove that $\widehat{X}$ satisfies the conditions for consistency and asymptotic normality, and then "transfer" these properties to $X$, as defined in section \ref{sec: reg framework}, by showing that $\widehat{X}$ and $X$ are sufficiently close. 

\begin{theorem} \label{theorem: suff cond for consistency}
Let $B \in \R^{(3p+1) \times (p+1)}$ be as defined in section \ref{sec: reg framework}, and $B_v = \vecop\fr{B}$. Define: 
\begin{align*}
    \Lambda_g &= \set{i \in V}{D^*_i \geq \sqrt{\sigma}n} \text{, $\sigma \in (0, 1)$}\\
    \widehat{\alpha}_i &= \exp\fc{\widehat{X}_{i*}^TB}\\
     \widehat{\Sigma}_i &= \diag\fr{\psi^{(1)}\fr{\widehat{\alpha}_i}} - \psi^{(1)} \fr{\textbf{1}_{p+1}^T\widehat{\alpha}_i }.
\end{align*}
In addition, assume $\frac{1}{n}D^*_i$ has a density function, $f$, that satisfies $f(x) \leq k_bx^{-\delta_b}$ on $(0, 2\sqrt{\sigma})$, for some $\delta_b < 1$, and $k_b > 0$. Consider the following Dirichlet GLM: $Z_{i, t+1} \sim \Dir\fr{\exp\fc{X_{i*}^T B}}$.\\

 If the following conditions hold:
\begin{enumerate}
    \item $\sigma \in \omega(n^{-\frac{1}{2}}) \cap o(1)$
    \item $\lambda_{\min}\E\fr{\fr{\widehat{X}_{i*}\otimes I_{p+1}}\widehat{\Sigma}_i\fr{\widehat{X}_{i*}^T\otimes I_{p+1}}} > 0$,
\end{enumerate}

 then, almost surely, the MLE of $B_v$, $\widehat{B}_v$, asymptotically exists and it is consistent and asymptotically normal. 
\end{theorem}

\begin{remark}
    Note that the second condition is reasonable because $\widehat{A}^w_{i}$ is a not a linear function of $Z_i$,
\begin{align*}
    \widehat{A}^w_i = \widehat{N}_i \widehat{D}^{-1}_i = \fc{\E\fr{Z_{\tau_w(i)}^TZ_{\tau_w(i)}}Z_i} \fc{\E\fr{\textbf{1}_{\abs{\tau_w(i)}}^TZ_{\tau_w(i)}} Z_i}^{-1}.
\end{align*}
Hence, $\hat{X}$ is not necessarily rank-deficient. \\
\end{remark}

\begin{corollary}
    Let $T_0 = I_3 \otimes \vecop^T\fr{\begin{bmatrix}
        I_{p}\\
        -\mathbf{1}^T_{p}
    \end{bmatrix}}$. Define $C \in \R^{(3p+1)(p+1) \times 4}$ to be 
    \begin{align*}
        C = \begin{bmatrix}
            T & \mathbf{0}_{3p(p+1)}\\
            0_{(p+1) \times 3} &\mathbf{1}_{p+1}\\
            \mathbf{1}^T_{3} & 1
        \end{bmatrix}
    \end{align*}
    Since $B_v = C\beta$, the MLE of $\beta$ will be given by $\widehat{\beta} = (C^TC)^{-1}C^TB_v$. 
\end{corollary}


\subsection{Estimate is consistent with latent positions aligned by an oracle}

 In this section, we present results for consistent estimation of the regression parameters when only the networks are observed.
Specifically, if an oracle is used to resolve the non-identifiability issue from ASE, then the MLE obtained using ASE is still consistent. To show this consistency, we use the consistency of ASE\cite{RDPG_Survey}, together with the implicit function theorem\cite{munkres1991analysis}. Before proceeding to the next theorem, let $\delta \in (0, 1)$, $\phi_{\delta}: \R^p \to D_p(\delta)$ be the orthogonal projection to $D_p(\delta)$, and we shall define the following:
\begin{align*}
    D_{p}(\delta) =& \set{Z \in \R^{p}}{Z^T \textbf{1}_{p} \leq 1 - \delta \text{ and } \min_{j \leq p}Z_j > \delta}.
\end{align*}

\noindent We first state a theorem that applies to any $2\to\infty$ consistent estimate for the latent positions.
\begin{theorem} \label{theorem: oracle align consistency}
    Under the same settings and assumptions of Theorem \ref{theorem: suff cond for consistency}, let $\widehat{Z}_0, \widehat{Z}_1$ be estimates of $Z_0, Z_1$, respectively, and suppose the following conditions hold:
    \begin{enumerate}
        \item $\norm{Z_s - \widehat{Z}_s}_{2 \to \infty} = O_p(\epsilon)$ for $s = 0, 1$, and $\lim\limits_{n \to \infty} \epsilon(n) = 0$
        \item $\max\limits_{i \leq n, j \leq p+1}\exp\fc{X_{i*}^TB_{*j}} > 2$,
    \end{enumerate} 
    then for $s = 0, 1$, the estimate $\Tilde{B}$ obtained using $\Tilde{Z}_s= \phi_{\epsilon}\fr{\widehat{Z}_s}$ ($\phi$ is applied row-wise) instead of $Z_s$ satisfies:
    \begin{align*}
        \norm{\Tilde{B} - \widehat{B}}_2 = O_p(\epsilon) \text{ where $\widehat{B}$ is the MLE obtained using $Z_s$}.
    \end{align*}
\end{theorem}

\noindent The following corollary holds as a direct consequence of Theorem \ref{theorem: suff cond for consistency} and Lemma \ref{lemma: ASE premise}.
\begin{corollary}
    Consider the adjacency matrices $Y_s$ for $s = 0, 1$. Let $\widehat{Z}_s$ be the ASE of $Y_s$. There exists $W_s \in O_p$ such that $\widehat{Z}_s W_s$ is a consistent estimator of $Z_s$, the true latent position. In addition, the MLE computed using $\phi_\epsilon\fr{\widehat{Z}_sW_s}$ converges almost surely to $\widehat{B}$, the MLE computed using $Z_s$, the true latent position. Here $\epsilon = \frac{C \log^2(n)}{\sqrt{n}}$ for some $C \in \R+$. 
\end{corollary}

 In the following section, we show from numerical simulation that there exists cases where we can estimate the parameters of our model as accurately with or without oracle. However, it will be our future work to quantify these conditions and prove the corresponding theoretical results.

\subsection{Numerical Simulations}
 We conducted Monte-Carlo simulations to assess our estimator. The settings are as follows:
\begin{enumerate}
    \item $K$, the number of groups, is equal to $3$.
    \item $p$, the embedding dimension, is $2$.
    \item The initial latent position are sampled from a mixture of Dirichlet distributions, with parameter $(1, 1, 10), (1, 10, 1), (10, 1, 1)$ and equals weights for each mixing component. See Figure \ref{fig: Align Methods}.
    \item $n$, the number of nodes, ranges from $1500$ to $12000$ with an increment of $1500$.
    \item $\beta$, the regression coefficients that we are estimating is $[1, 1, -4, 5]$. See Figure \ref{fig: Polarization Exp}.
\end{enumerate}

 In Figure~\ref{fig: node vs. bias}, the estimate using the oracle latent positions has approximately $0$ bias. In addition, whether we align ASE using an oracle or RGD, the estimate is biased, but the distribution of the bias is roughly the same for both cases, indicating this bias is due mainly to the fact that the latent positions are estimated and not due to alignment issues.
The estimate of $\beta_4$ is the most inaccurate by far, but it this is usually acceptable because it is often more of a nuissance parameter representing overall variance of the latent positions, and we are more interested in estimating the other $3$ parameters. 
In both cases, the bias and standard deviations of the estimates decrease as $n$ increases.
\\
\begin{figure}[H] \label{fig: decreasing bias}
    \centering
    \includegraphics[width=0.8\textwidth]{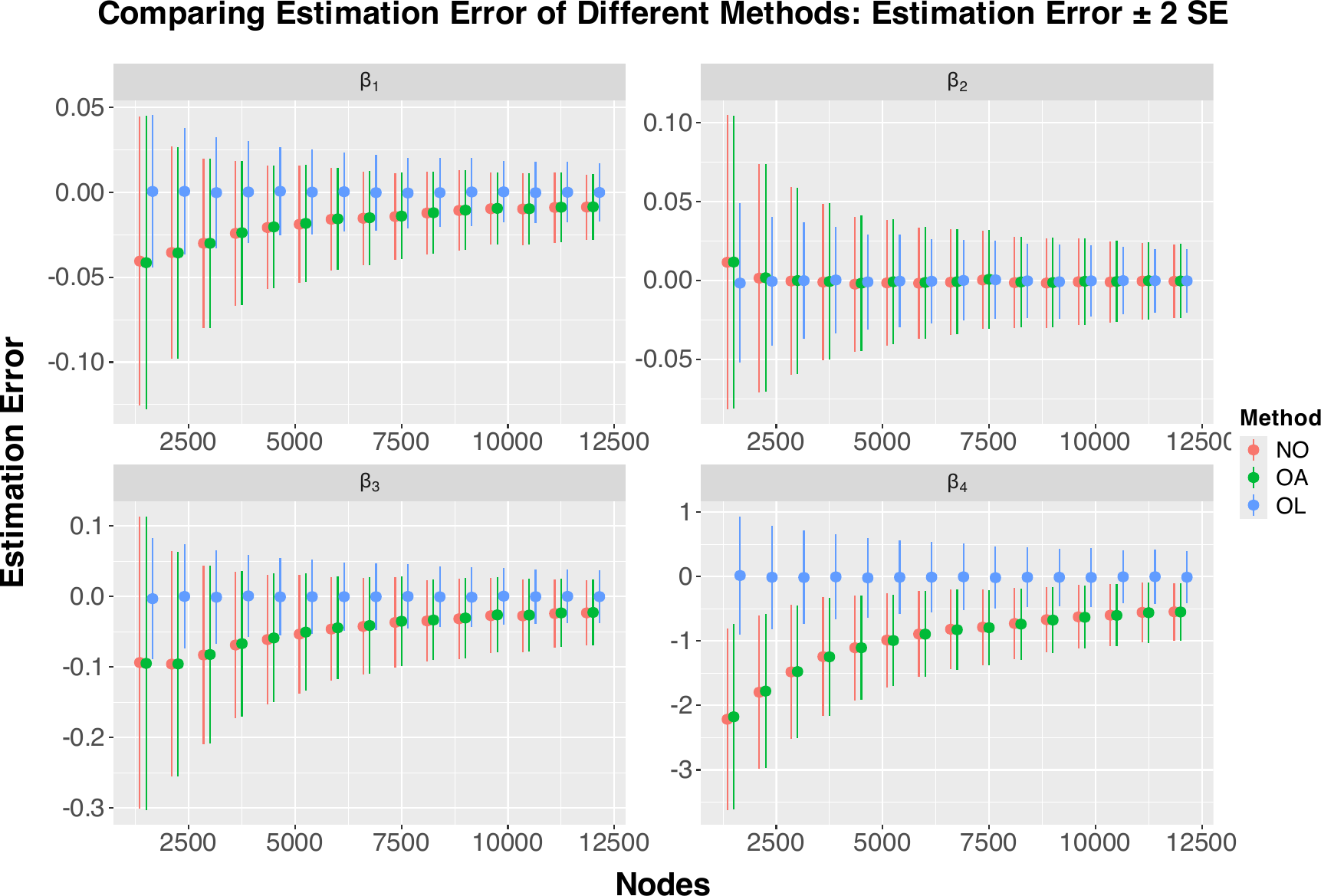}
    \caption{This is a plot of the number of nodes vs.
    Mean $\pm\ 2$ SD of the estimation error of each component of $\beta$. Different colors are used to distinguish the method of estimation: NO is the "no-oracle" method, OA is the "oracle-alignment" method, and OL is the "oracle-latent-position" method.}
    \label{fig: node vs. bias}
\end{figure}

 In Figure~\ref{fig: emp vs. theo. SD}, we are checking the efficiency of our estimate by comparing the standard deviation of our estimates to the theoretical standard deviation predicted by the GLM theory. With oracle latent positions, this ratio is very close to $1$ for all parameters, thus supporting our theory that our estimate is normal. With the other two methods, the ratios for $\beta_1, \beta_2$ are close to $1$. For $\beta_3$, it is not as close, but trends toward $1$. As for $\beta_4$, it is the least accurate, but it also trends toward $1$.

\begin{figure}[H]
    \centering
    \includegraphics[scale = 0.6]{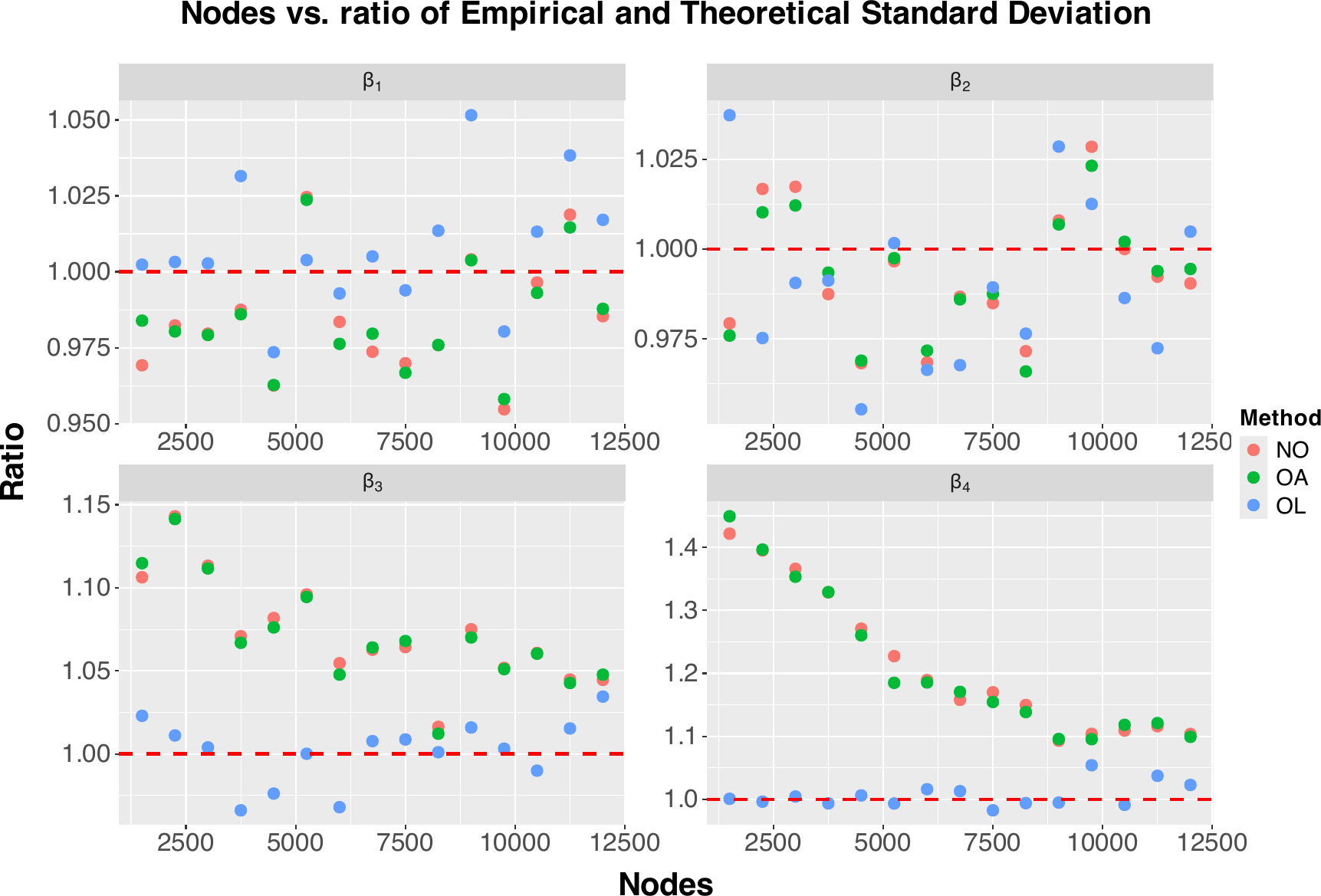}
    \caption{This is a plot of the number of nodes vs. the ratio of empirical and theoretical SD for each component of $\beta$. The color code is identical to that of Figure \ref{fig: decreasing bias}.}
    \label{fig: emp vs. theo. SD}
\end{figure}

\section{Real Data} \label{sec: real data}
We shall examine a network representing online computer game (Age of Empires IV, AOE4 \footnote{AOE4 is a real-time strategy game where players manage civilizations, and build armies to engage in warfare. In a 1v1 match, players win by fulfilling some victory conditions that represents dominance over their opponent}) matches to assess the ability of our method to capture flocking and polariation behavior in real data. We construct the network and groups to capture flocking/polarization behaviors that are partially built into the online match-making system. We will use the match data\footnote{The data is provided by aoe4world under Microsoft's "Game Content Usage Rules" using assets from Age of Empires IV, and it is not endorsed by or affiliated with Microsoft. The specific data sets can be found at \url{https://aoe4world.com/dumps}.}  of 1v1 ranked matches from 02/17/2023 to 03/19/2024. Each match in the dataset involves two players of similar skill levels. \\

In the dataset, each row represents a unique match. Some relevant variables include the date of the match, player-id, and matchmaking rank (MMR) for both players. Player-id uniquely identifies each AOE4 player. Players gain/lose MMR after winning/losing each ranked match. MMR will be used as an indicator for the level of skill of a player. This data set is naturally a time series of edges. Each node is a unique player, and an edge between two nodes means that the two players played at least one game over some pre-specified period of time. As the popularity of the game increase/decrease over time, players will join/leave the network, and the connectivity of the network will also increase/decrease. \\

 We created networks for two time disjoint time intervals, denoted period 0 at $t=0$ and period 1 at $t=1$. The period $0$ network is the match network from Feb. 17, 2023 to Oct. 7, 2023, and period $1$ network is the match network from Oct. 8, 2023 to Mar. 19, 2024. The time periods were chosen to ensure that there are roughly the same number of matches in both networks (2,255,507 and 2,254,826 matches, respectively). The full networks have 112,758 and 118,174 nodes at period $0$ and $1$ correspondingly. Since our model is about detecting and quantifying polarizing/flocking behavior in a network, we constructed two groups from the data that should display these behaviors, i.e. two groups where the connectivity between them decreased/increased when going from $t = 0$ to $t = 1$. 
One natural way based on the mechanism of matchmaking is to look at low-skilled players who got worse at the game vs. high-skilled players who got better at the game (polarizing), and low-skilled players who got better at the game vs. high-skilled players who got worse at the game (flocking). \\

 We calculated the mean MMR for each player during each period and used this to define two binary attributes: MMR-group and trend-group. The MMR-groups $0$ and $1$ represent players whose mean MMR during period $0$ is below or above the median of the mean MMRs, respectively. Trend-group $0$ includes players whose change in mean MMR from period $0$ to period $1$ is below the median of these changes, while trend-group $1$ includes those above it. Each player is characterized by an ordered pair $(\text{MMR-group}, \text{trend-group})$, representing these attributes. 
The networks formed from groups $(0,0)$ and $(1,1)$, which we denote the "away graph", are expected to exhibit polarizing behavior, while networked formed from groups $(0, 1)$ and $(1, 0)$ are expected to display flocking behavior (the "toward graph"). \\

To reduce the sparsity of our network, we filtered out players who played fewer than $50$ games in each period.
After filtering, there are $7552$ players in total. Here we will mainly focus on the away graph. $(0,0), (1, 1)$ both have $1833$ players. We present some basic information for the away graph below, MD stands for median degree, and MD-BW is the median number of connections from $(0,0)$ to $(1, 1)$. 
    
\begin{table}[H]
    \centering
    \begin{tabular}{|c|c|c|c|c|c|}
        \hline
        Period & $\abs{E}$ & MD-Overall & MD-$(0, 0)$ & MD-$(1, 1)$  & MD-BW\\
        \hline
        $0$ & $92323$ & $40$ & $31$ & $38$ & $1$\\
        \hline
        $1$ & $61901$ & $25$ & $18$ & $33$ & $0$\\
        \hline
    \end{tabular}
    \caption{Basic Info for the Away Graph}
    \label{table: away graph basic info}
\end{table}

\subsection{The "Away Group"} \label{sec: the away group}
 In Figure~\ref{fig: Away Adj}, the rows and columns of the adjacency matrices are sorted by the mean MMR in period $0$. The left plot confirms that players are matched primarily with other players with similar MMR.
A comparison with the right plot reveals the polarizing behavior, with two groups whose players compete in very few matches between each other.
\begin{figure}[H]
    \centering
    \includegraphics[width = 0.8\textwidth]{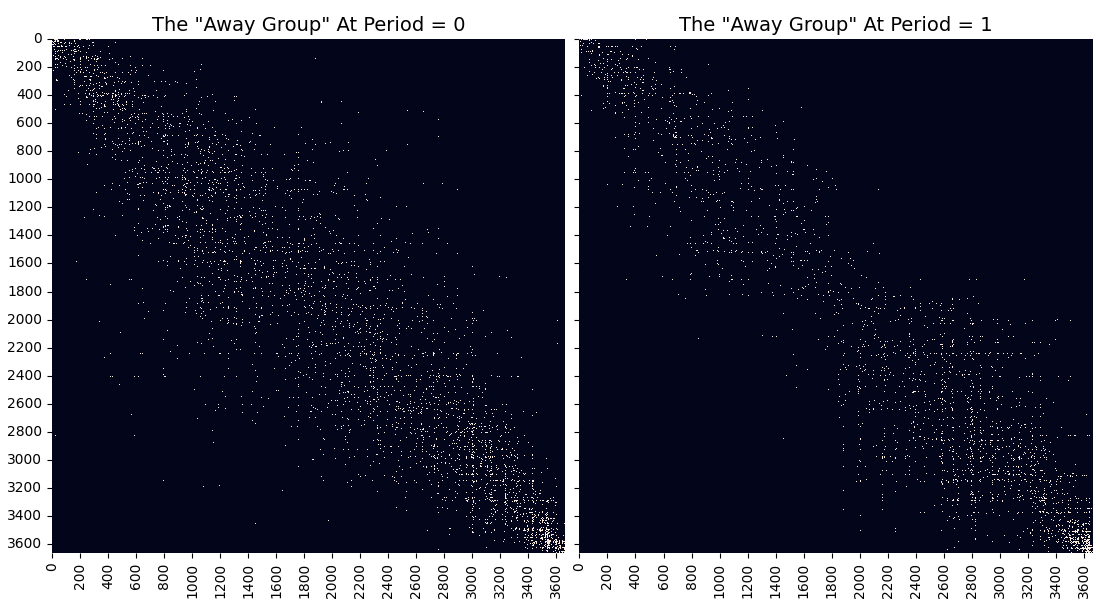}
    \caption{Here are the adjacency matrices for the Away Group at period $0$ and $1$. Rows and columns represent players in the group, sorted by MMR rank as indicated on the axis labels. The two MMR groups, are splitted at roughly rank 1800. Compared to the adjacency plot at period $1$, we see that there are a lot more connections between the two MMR groups at period $0$. }
    \label{fig: Away Adj}
\end{figure}

 We embedded the graphs in $\R^5$ based on Figure~\ref{fig: Eigen Away} in Appendix. After aligning the two networks, the first two dimensions of the estimated latent positions are shown in Figure~\ref{fig: LP Away}.

\begin{figure}[H]
    \centering
    \includegraphics[width=0.8\textwidth]{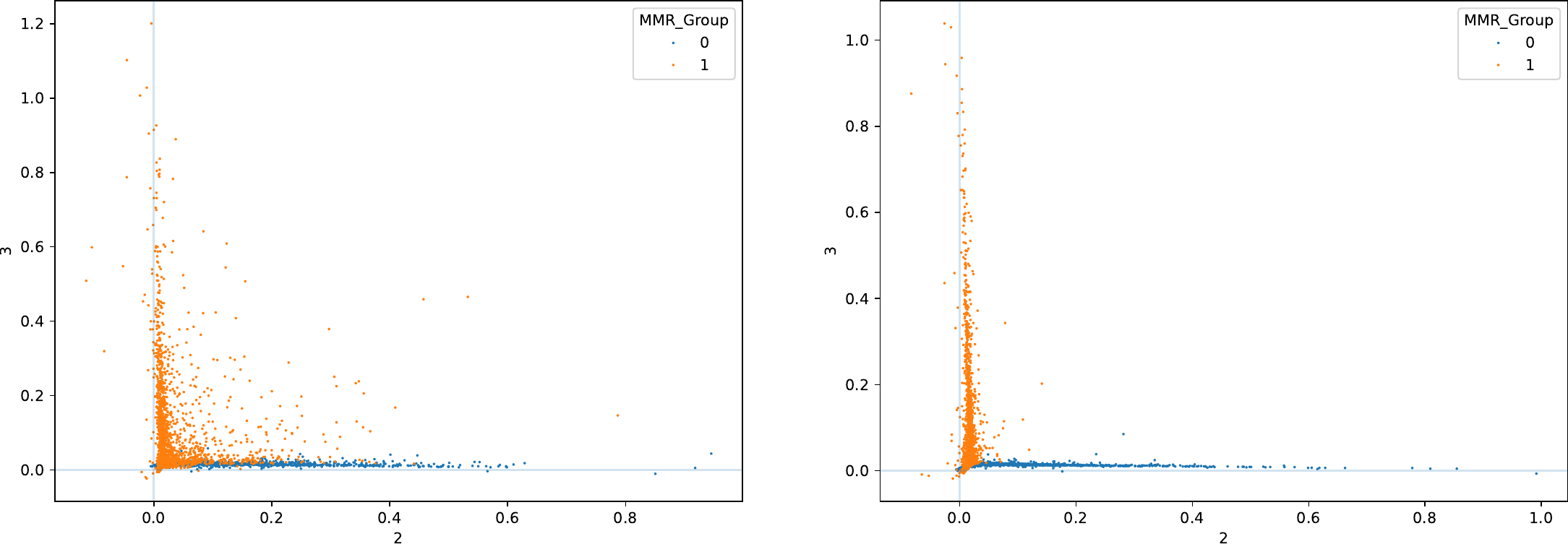}
    \caption{The plots above are the canonical projections of the estimated latent positions via GAEP (from $\R^5$) to $\R^2$. We can see that through penalization, most of the latent positions are inside $\Delta^5$. At period $0$, the latent positions look perpendicular from afar, but there are a lot of interactions at the "angle", which correspond to players ranked around 1800. At period $1$, the latent positions of the two groups are still connected, but the interaction at the connection visibly decreased by a lot, showing the predicted polarizing behavior.}
    \label{fig: LP Away}
\end{figure}
 We see in Figure~\ref{fig: LP Away} that the latent positions of the two groups become more separated when going from period $0$ to period $1$. Fitting our model to this data, our estimate for $\beta$ can be found in Table \ref{tab: Result Est Away - 5}. As mentioned previously, the $\beta$'s represent the different forces that drive the dynamics. Similar value of $\beta_1$ and $\beta_2$ shows that for each node, the force that its own latent position and the within-group attractor exerts are very similar. $\beta_3 = -0.41$ indicate that $(0, 0)$ and $(1, 1)$ are repelling each other, albeit weakly.
If we were to test the null hypothesis that $\beta_3 = 0$ vs. $\beta_3 < 0$, then we would likely be rejecting the null hypothesis judging by the theoretical standard deviation. This is consistent with our hypothesis that in the subgroup of players that we constructed, polarization is happening. Judging by Figure~\ref{fig: Eigen Away}, the scree plot for the adjacency matrix, embedding the data in $\R^5$ is one of the reasonable choices. 

\begin{table}[h]
    \centering
    \begin{tabular}{|c|c|c|c|c|}
        \hline
         & $\beta_1$ & $\beta_2$ & $\beta_3$ & $\beta_4$\\
         \hline
        Estimate & $1.5946$ & $1.6428$ & $-0.4141$ & $1.1258$\\
        \hline
        Theoretical St.Dev. & $0.0357$ & $0.0594$ & $0.1258$ & $0.0854$\\
        \hline
    \end{tabular}
    \caption{Estimated Parameters and Their Theoretical Standard Deviation}
    \label{tab: Result Est Away - 5}
\end{table}
 Embedding the adjacency matrices in $\R^5$, result of our model is consistent with our expecation, i.e. the network that we constructed is polarizing from period $0$ to period $1$. As for other embedding choices, we see in Figure~\ref{fig: Dim vs Est Away} that besides $\R^2$, all other embedding choices (up to $\R^9$) yields similar results, suggesting there is some robustness of the estimates to dimension misestimation.

\begin{figure}[h]
    \centering
    \includegraphics[scale = 0.4]{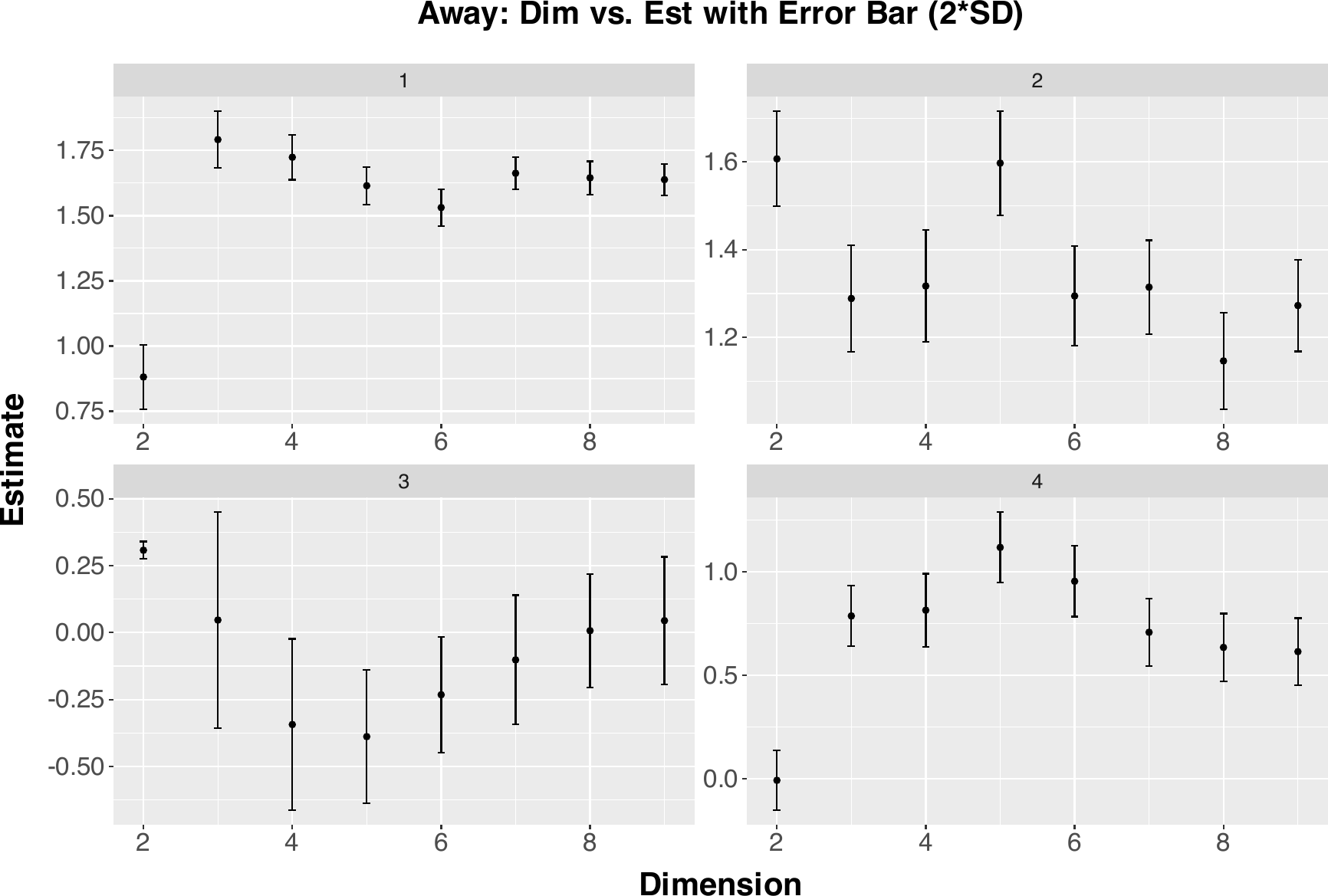}
    \caption{This is the the plot of embedding dimension vs. estimate of components of $\beta \pm 2$SD. Although there are outliers, we see that the estimates are similar from $\R^3$ to $\R^9$. As for the indicator of polarization, estimates of $\beta_3$ are mostly negative if not very close to $0$, which aligns with our expectation. Overall we see some robustness to dimension mis-specification through this data study.}
    \label{fig: Dim vs Est Away}
\end{figure}

\subsection{The "Toward Group"}
 Unlike the Away group, at period $1$, since the players' MMR in the two groups are closer, more games happened between the two groups resulting in the expected flocking behavior. We fit the our model in $\R^2, \R^3$ up to $\R^9$, and the plot below shows our estimation of each component of $\beta$ vs. the number of embedding dimensions with $\pm 2$ SD. 
\begin{figure}[H]
    \centering
    \includegraphics[scale = 0.4]{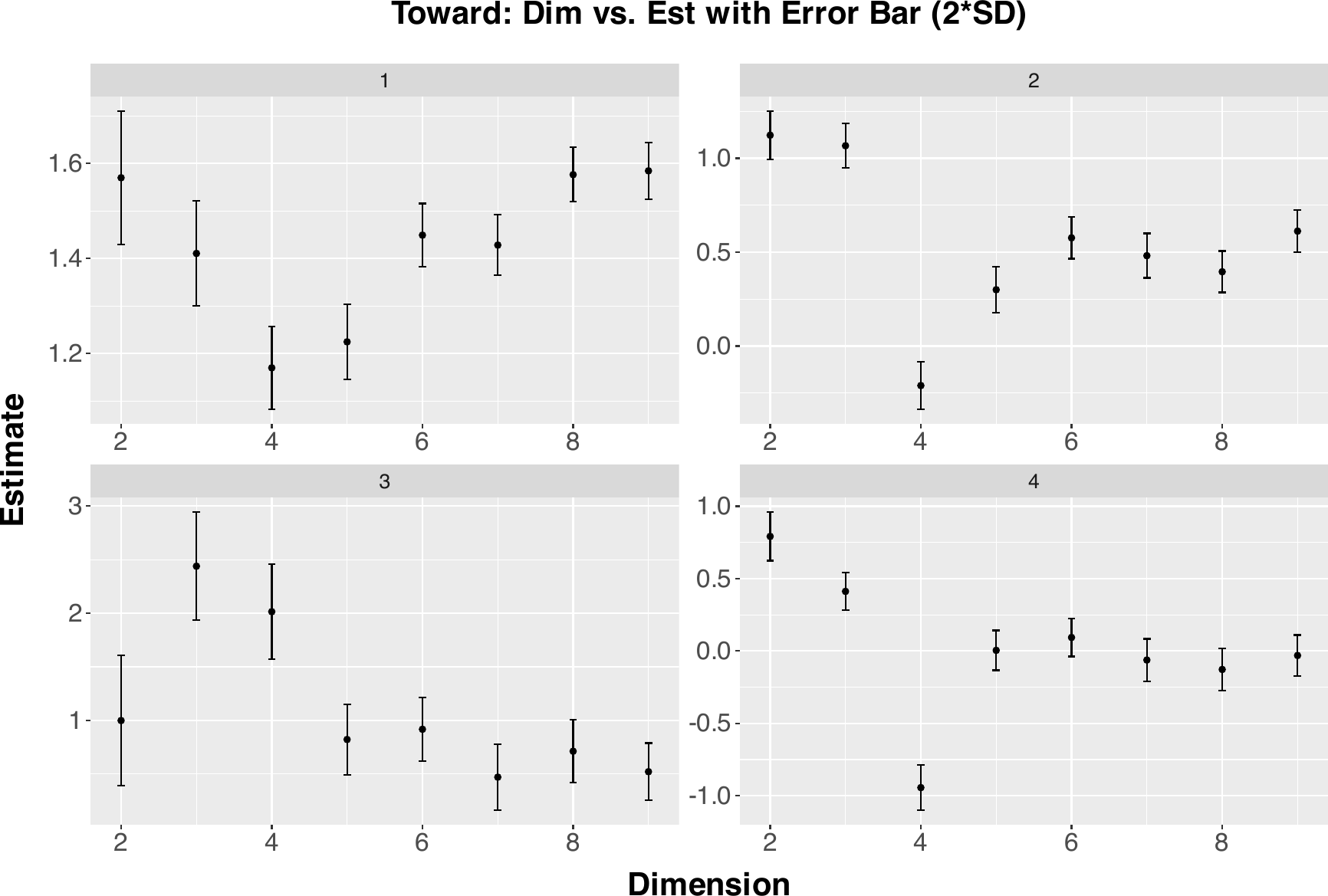}
    \caption{This is the the plot of embedding dimension vs. estimate of components of $\beta \pm 2$SD. Although there are outliers, we see that the estimates are similar from $\R^3$ to $\R^9$. As for the indicator of polarization, estimates of $\beta_3$ are very positive, distant from $0$. This aligns with our expectation of the Toward group that there will be flocking behavior. Again we see some robustness to dimension mis-specification.}
    \label{fig: Dim vs Est Toward}
\end{figure}

 Overall, our $\beta$ estimate using different number of embedding dimensions is relatively stable. The major difference between this and the estimate for the Away group is that we are expecting a positive $\beta_3$ because of the flocking behavior, and our estimations above confirm exactly this.

\section{Disucssion} \label{sec: discussion}


\noindent Inspired by the CLSNA model, we developed Attractor-Based Coevolving Dot Product Random Graph Model (ABCDPRGM), a random dot product graph version of the coevolving latent space model. We aim to model the polarization/flocking behavior of a multiple communities, and, by specifying the parameters of each attractor, we can control the rate of polarization/flocking. The main inferential task for this model is to estimate the parameters of each attractor, which involves first estimating the latent positions through ASE and then using the estimated latent positions to fit a Dirichlet GLM. We have shown that our estimate is consistent under some oracle conditions.\\

\noindent In the original CLSNA model, estimating the latent positions requires Markov Chain Monte Carlo(MCMC), which is very time-consuming. Later, improvements were made via Stochastic Gradient Descent(SGD)\cite{pan2024stochastic}, and latent position estimation for CLSNA became much faster. However, our model is still much faster because as a RDPG-based model, recovering the latent positions (using ASE) only requires computing a partial SVD of the adjacency matrix. \\

\noindent One limitation of our model is that we are asuming that the set of nodes does not change with respect to time, thus leaving the nodes in the network that come and go unaccounted for. For example, in the AOE IV data set, we included all players who played in both period $0$ and $1$ to our network, but there are plenty of other players who played in only one of the periods. This will be an interesting direction to generalize our model to accomodate networks with varying set of nodes. \\

\noindent In addition, using mixed membership where each node can have partial membership to multiple groups is another future direction to generalize our model. Currently in our model, each node belongs to exactly one group, but this is often not the case in reality. For example, looking at a friendship network, if we define group membership based on beliefs about gun control, very few people will be totally for or against gun control. Instead, people will be scattered on a spectrum ranging from "totally for gun control" to "totally against gun control". Mixed membership models like the one introduced in \cite{JIN2024105369} allows us to account for this.\\

\noindent One limitation of our theory is related to the identifiability problem of RDPG. So far, we have shown consistency of our estimates if the identifiability problem is addressed by some oracle. Without oracle, we proposed a loss function to adress the identifiability problem. We have found setups where doing gradient descent with this loss function works very well. Our future research will focus on better quantifying these conditions, and proving consistency results with these conditions.\\

\noindent In our analysis of the AOE IV data set, we constructed two groups of "polarizing" players to check if our model is able to detect the polarization. Embedding the data in $\R^5$, we confirmed that our model was able to detect the "polarization". However, since there is no obvious correct choice for $p$, the embedding dimension, we tried a wide range of reasonable choices (as shown in Figure \ref{fig: Dim vs Est Away}). It became clear to us that while there are fluctuations as we change the embedding dimensions, the estimates are all very similar, thus demonstrating some level of robustness for misspecified embedding dimensions. \\

\noindent In this article, we introduced ABCDPRGM, methods to estimate the parameters of ABCDPRGM, proved consistency for our estimates under certain conditions, and analyzed a real data set. While the assumptions of our model can be a bit strict, e.g. latent positions being in the simplex, we proposed methods to apply our model to cases where assumptions of our model fails to hold, and demonstrated some level of robustness through these cases. In future work, we plan to make our model more flexible by incorporating mixed membership, and expand on the theory about the no-oracle case.

\bibliographystyle{plain}
\bibliography{refs}

\appendix
\section{Proof of Main Results}
\subsection{Theorem \ref{theorem: suff cond for consistency}}
We first prove Theorem 2 which establishes the consistency of the maximum likelihood estimator under our model when the latent positions are observed.

\begin{proof}
     In this proof, we will show that under the assumptions of Theorem \ref{theorem: suff cond for consistency}, the following conditions, as discussed in Section~\ref{subsec: GLM consistency} are satisfied. 
    
    \begin{itemize}
    \item[(D)] Divergence: $\lambda_{min}\fc{F_n} \to \infty$
    \item[(N)] Convergence and Continuity: $\forall \delta > 0,\ \max_{\Tilde{B} \in N_n(\delta)} \norm{V_n(\Tilde{B}) - I} \to 0$, where $V_n(\Tilde{B}) = F_n^{-1/2}H_n(\Tilde{B})F_n^{-T/2}$
    \item[($S$)] Boundedness of the eigenvalue ratio: $\exists \text{ neighborhood } N$ of $B$ s.t. 
    \begin{align*}
        \lambda_{\min}\fc{H_n(\Tilde{B})} \geq c(\lambda_{\max}\fc{F_n}), \text{ with $\Tilde{B} \in N, c, \delta > 0$, and $n$ sufficiently large}
    \end{align*}
\end{itemize}
    
 \textbf{Settings:} We start by restating the following computation: 
\begin{align*}
    \ell\fr{B| Z^*} &= \s{i = 1}{n} \alpha_i^T \log(Z^*_{i*}) - \fr{ \textbf{1}_{p+1}^T \log(\Gamma\fr{\alpha_i} - \log\fr{ \Gamma\fr{ \textbf{1}_{p+1}^T \alpha_i}} - \textbf{1}_{p+1}^T \log\fr{Z^*_{i*}}}\\
    s_n(X, B) &= \pderiv{B_v}{\ell\fr{B| Z^*}}=  \s{i = 1}{n} \fr{X_{i*} \otimes I_{p+1}} \diag \fr{\alpha_i}\fr{\log(Z^*_{i*}) - \mu_i(\alpha_i)}:= \s{i = 1}{n}s(X_{i*}, B)\\
    F_n(X, B) &= \s{i = 1}{n} \fr{X_{i*} \otimes I_{p+1}} \diag\fr{\alpha_i}\Sigma_i(\alpha_i)\diag\fr{\alpha_i} \fr{X_{i*}^T \otimes I_{p+1}} := \s{i = 1}{n}F(X_{i*}, B)\\
    R_n(X, B) &= \s{i = 1}{n} \fr{X_{i*} \otimes I_{p+1}} \diag\fr{\fb{\log(Z^*_{i*}) - \mu_i(\alpha_i)}\circ \alpha_i}\fr{X^T_{i*} \otimes I_{p+1}}:= \s{i = 1}{n}R(X_{i*}, B),
\end{align*}
where 
\begin{align*}
    \alpha_i &= \exp\fc{\fr{X_{i*}^T \otimes I_{p+1}}B_v} = \exp \fc{X_{i*}^T B},\\
    \mu_i(\alpha_i) &= \psi\fr{\alpha_i} - \psi \fr{\textbf{1}_{p+1}^T\alpha_i },\\
    \Sigma_i(\alpha_i) &= \diag\fr{\psi^{(1)}\fr{\alpha_i}} - \psi^{(1)}\fr{\textbf{1}_{p+1}^T\alpha_i }.
\end{align*}
    \noindent $\psi$ and $\psi^{(1)}$ are the digamma and trigamma function defined to be the first second derivative of the log-gamma function. Now we proceed to verify the three conditions.

     \textbf{Condition (D)}: We need to show that, almost surely, $\lambda_{\min} F_n(X, B) \to \infty$. We first show that $\lambda_{\min} F_n(\widehat{X}, B) \to \infty$ through an LLN argument, and then bound the distance between $F_n(X, B)$ and $F_n(\widehat{X}, B)$. For $\nu \in \R^{3p+2}$:
\begin{align*}
    \frac{1}{n}\nu^T F_n\fr{\widehat{X}, B} \nu &= \frac{1}{n}\nu^T\fr{\s{i = 1}{n}(\widehat{X}_{i*}\otimes I_{p+1}) \diag(\widehat{\alpha}_i) \widehat{\Sigma}_i \diag(\widehat{\alpha}_i) (\widehat{X}_{i*}^T\otimes I_{p+1})}\nu\\
    &\geq \frac{k_0^2}{n} \fr{\s{i=1}{n}\nu^T\fr{\widehat{X}_{i*}\otimes I_{p+1}}\widehat{\Sigma}_i\fr{\widehat{X}_{i*}^T\otimes I_{p+1}}\nu } \quad \text{let $\min_{ij} \widehat{\alpha}_{ij} = k_0 > 0$}\\
    &\asto k_0^2 \fr{ \nu^T \E\fr{\fr{\widehat{X}_{i*}\otimes I_{p+1}}\widehat{\Sigma}_i\fr{\widehat{X}_{i*}^T\otimes I_{p+1}}} \nu}\\
    &\geq k_0^2\norm{v}^2_2 \lambda_{\min} \E\fr{\fr{\widehat{X}_{i*}\otimes I_{p+1}}\widehat{\Sigma}_i\fr{\widehat{X}_{i*}^T\otimes I_{p+1}}} \\
    &> 0 \text{ if $\nu \neq 0$}.
\end{align*}
    \noindent Next we bound the distance. Deine $G_F(\xi) = \given{\pderif{R}{F(R, B)}}_{R = \xi}$. Recall that $\Lambda_g = \set{i \in V}{D^*_i \geq \sqrt{\sigma}n}$, and $\Lambda_b = V - \Lambda_g$:
    \begin{align*}
        &\frac{1}{n} \norm{F_n(X, B)- F_n\fr{\widehat{X}, B}}_2 \\
        =& \frac{1}{n} \norm{\s{i = 1}{n} F(X_{i*}, B) - F\fr{\widehat{X}_{i*}, B}}_2\\
        =& \frac{1}{n} \norm{\s{i = 1}{n} G_F\fr{X^*_{i*}} \fr{X_{i*} - \widehat{X}_{i*}}}_2.\\
    \end{align*}
    This is due to Taylor's theorem. Here $X^*_{i*}$ is a point on the line segment connecting $X_{i*}$ and $\widehat{X}_{i*}$. Next we split the indices into $\Lambda_g$ and $\Lambda_b$, and bound the norm separately:
    \begin{align*}
        &\frac{1}{n} \norm{F_n(X, B)- F_n\fr{\widehat{X}, B}}_2 \\
        =& \frac{1}{n} \norm{\s{i \in \Lambda_g}{} G_F\fr{X^*_{i*}} \fr{X_{i*} - \widehat{X}_{i*}} + \s{i \in \Lambda_b}{} G_F\fr{X^*_{i*}} \fr{X_{i*} - \widehat{X}_{i*}}}_2\\
        \leq& \norm{X_{\Lambda_g} - \widehat{X}_{\Lambda_g}}_{2 \to \infty}\frac{1}{n}\s{i \in \Lambda_g}{}\norm{G_F\fr{X^*_{i*}}}_2 + \norm{X_{\Lambda_b} - \widehat{X}_{\Lambda_b}}_{2 \to \infty}\frac{1}{n}\s{i \in \Lambda_b}{}\norm{G_F\fr{X^*_{i*}}}_2\\
        =& o_p(1).
    \end{align*}
    \noindent The bound above holds because of the following:$\norm{X_{\Lambda_g} - \widehat{X}_{\Lambda_g}}_{2 \to \infty} = o_p(1)$ by Lemma \ref{lemma: X_lambda_g bound}, $\norm{X_{\Lambda_b} - \widehat{X}_{\Lambda_b}}_{2 \to \infty} = O(1)$ by definition, $\abs{\Lambda_b} = o_p(1)$ by Lemma \ref{lemma: size of Lambda b}, and $\max\limits_{i \in V}\norm{G_F\fr{X^*_{i*}}}_2 < M$ for some $M \in \mathbb{H}$, since $G$ is continuous, and $X_{i*}$ is on a compact set for all $i \in V$. \\
    
     Since $F_n(X, B), F_n(\widehat{X}, B)$ are close enough, and $\lambda_{\min} F_n(\widehat{X}, B) \to \infty$, we get $\lambda_{\min} F_n(X, B) \to \infty$ as desired.\\

     \textbf{Condition (N)}: For $\delta>0$, let $N_n(\delta) = \set{\Tilde{B} \in \R^{(3p+1)\times (p+1)}}{\norm{F_n^{T/2}(X, \Tilde{B})\fr{B - \Tilde{B}}}_2\leq \delta}$. We need to show that, almost surely, for all $\delta, \epsilon > 0$, there exists $n_1 > 0$ such that for all $n > n_1$:
    \begin{align*}
        \max_{\Tilde{B} \in N_n(\delta)} \norm{F_n^{-1/2}(X, \Tilde{B})H_n(X, B)F^{-T/2}_n(X,\Tilde{B}) - I_n}_2 &< \epsilon \quad \text{where $H_n(X, B) = F_n(X, B) + R_n(X, B)$},\\
        \text{or equivalently\ }\max_{\Tilde{B} \in N_n(\delta)}\frac{1}{n}\norm{H_n(X, B) - F_n(X, \Tilde{B})}_2 &< \epsilon \quad \text{since $\norm{F_n^{-1}(X, B)}_F = O\fr{n^{-1}}$}.
    \end{align*}
    
    \noindent Define $G_R(\zeta) = \given{\pderif{U}{R(U, B)}}_{U = \zeta}$ Let $\Tilde{B} \in N_n(\delta)$, then:
    \begin{align*}
        &\frac{1}{n}\norm{H_n(X, B) - F_n(X, \Tilde{B})}_2\\
        \leq& \frac{1}{n}\fr{\norm{F_n(X, B) - F_n(X, \Tilde{B})}_2 + \norm{R_n(X, B)}_2} \text{ let $B^*$ be the point between $B$ and $\Tilde{B}$ from the MVT}\\
        \leq& \norm{\Tilde{B} - B}_2 \frac{1}{n}\s{i = 1}{n} \norm{\given{\pderif{S}{F(X_{i*}, S)}}_{S = B^*}}_2 + \frac{1}{n}\norm{R_n\fr{\widehat{X}, B}}_2 + \frac{1}{n}\norm{R_n(X, B) - R_n\fr{\widehat{X}, B}}_2\\
        \leq& o(1) + \norm{X_{\Lambda_g} - \widehat{X}_{\Lambda_g}}_{2 \to \infty}\frac{1}{n}\s{i \in \Lambda_g}{}\norm{G_R\fr{X^*_{i*}}}_2 + \norm{X_{\Lambda_b} - \widehat{X}_{\Lambda_b}}_{2 \to \infty}\frac{1}{n}\s{i \in \Lambda_b}{}\norm{G_R\fr{X^*_{i*}}}_2\\
        =& o_p(1).
    \end{align*}
    \noindent The derivation of the bound above follows the exact same logic as the derivation of the similar bound for \textbf{Condition (D)}. 
    
     \textbf{Condition (S)}: We need to show that there is a neighborhood $N$ of $B$ such that for all $\Tilde{B} \in N$, $\frac{\lambda_{\min} H_n(X, B)}{\lambda_{\max}F_n(X, \Tilde{B})} \geq c > 0$ a.s. for $n \geq n_1$. From the proof of \textbf{condition (D), (N)}:
    \begin{align*}
        \frac{1}{n}\norm{H_n(X, B) - F_n(\widehat{X}, B)}_2 = o_p(1) \implies& \frac{1}{n}\fr{\lambda_{\min} H_n(X, B) - \lambda_{\min} F_n(\widehat{X}, B)} = o_p(1)\\
        \implies& \frac{1}{n}\fr{ \lambda_{\min} H_n(X, B) - n \lambda_{\min}\E\fr{F\fr{\widehat{X}_i, B}} } = o_p(1).
    \end{align*}
    \noindent By the same arguments, we have that:
    \begin{align*}
        \frac{1}{n}\fr{\lambda_{\max} F_n\fr{X, \Tilde{B}} - \lambda_{\max} \E\fr{F\fr{\widehat{X}_i, \Tilde{B}}}} = o_p(1).
    \end{align*}
    
    \noindent Combine everything above:
    \begin{align*}
         &\frac{\lambda_{\min} H_n(X, B)}{\lambda_{\max}F_n(X, \Tilde{B})} \geq \frac{\lambda_{\min} \E(F(X_i, B)) - \epsilon}{\lambda_{\max}\E(F(X_i, \Tilde{B}))+ \epsilon} \quad \text{$ \forall \epsilon > 0$, $n$ sufficiently large, uniformly for $\Tilde{B} \in N$ a.s.}.
    \end{align*}
    \noindent We have established $\lambda_{\min} \E(F(X_i, B)) > 0$ previously. As for $\lambda_{\max}\E(F(X_i, \Tilde{B})) > 0$, this holds since $\E(F(X_i, \Tilde{B})) \neq 0$. Therefore we have $\frac{\lambda_{\min} H_n(X, B)}{\lambda_{\max}F_n(X, \Tilde{B})} \geq c > 0$ as desired.\\
    
     With conditions \textbf{(D), (N), (S)}, the MLE of $B$, $\widehat{B}$ exists asymptotically, it is consistent and asymptotically normal almost surely. 
\end{proof}

\subsection{Theorem \ref{theorem: oracle align consistency}}

Now, we prove the consistency of the coefficient estimates when latent positions are estimated from the observed graph.

\begin{proof}
    Notation-wise, for convenience, we shall use the following in this proof:
    \begin{enumerate}
        \item $Z$ is in $\R^{n \times (p+1)}$ such that its row sum vector is a constant $1$ vector.
        \item $X$ is the design matrix from $Z_t$ 
        \item $Z$ will exclusively refer to $Z_{t+1}$ 
        \item any decorated version of $X, Z$ are defined analogously
        \item Any matrix with a subscript $v$ is its vectorized version, e.g. $B_v = \vecop(B), X_v = \vecop(X)$, etc. 
    \end{enumerate}

     We shall first invoke the implicit function theorem (IFT)\cite{munkres1991analysis}. In short, this theorem tells us that there is a unique continuously differentiable function, $g$, that maps data to MLE. Therefore small perturbation in data will translate to small perturbation in MLE. Recall that $B$ is the true parameter, $\widehat{B}$ is the MLE of $B$ using $X, Z$. Since $\widehat{B}_v$ is the root of $\pderiv{B_v}{\ell\fr{B_v; X, Z}}$, IFT states that if the Hessian of $\ell$ with respect to $B_v$ is invertible at $\widehat{B}_v$, i.e. $H_n^{-1}\fr{\widehat{B}_v; X, Z}$ exists, then:
    \begin{enumerate}
        \item There is an open set $U \subset \R^{n \times q} \times \R^{n \times \fr{p+1}}$ containing $(X, Z)$, where $q = 3p + 1 $.
        \item There is a unique continuously differentiable function $g: U \to \R^{q(p+1)}$ that satisfies the following conditions:
        \begin{enumerate}
            \item $g(X, Z) = \widehat{B}_v$,
            \item $\forall (X^*, Z^*) \in U,\ \pderiv{B_v}{\ell\fr{B^*_v; X^*, Z^*}} = 0$, where $B^*_v = g(X^*, Z^*)$.
        \end{enumerate}
    \end{enumerate}
     In addition, $\forall (X^*, Z^*) \in U$, $\left.\pderif{(R, S)}{g(R, S)}\right|_{\fr{R, S} = \fr{X^*, Z^*}}$ is characterized in the following way.:
    \begin{align*}
        \left.\pderif{(R, S)}{g(R, S)}\right|_{\fr{R, S} = \fr{X^*, Z^*}}  
        & = -\fr{\left. \pderif[2]{\Theta_v \partial \Theta_v^T}{\ell\fr{\Theta_v; X^*, Z^*}}\right|_{\Theta_v = g\fr{X^*, Z^*}}}^{-1} 
        \left. \pderif[2]{(R, S) \partial \Theta_v}{\ell\fr{\Theta_v; R, S}} \right|_{\Theta_v = g\fr{X^*, Z^*},\ \fr{R, S} = \fr{X^*, Z^*}}\\
        &= -H_n^{-1}\fr{g(X^*,Z^*); X^*, Z^*}\left.\pderif{(R_v, S_v)}{s_n\fr{\Theta_v; R, S}^T}\right|_{\Theta_v = g\fr{X^*, Z^*},\ \fr{R,S} = \fr{X^*, Z^*}}.
    \end{align*}

    \noindent Below is a list of notable values of $g$:
    \begin{enumerate}
        \item $\widehat{B}_v = g(X,Z)$, this is the true MLE from the true latent positions, $(X, Z)$.
        \item $\Tilde{B}_v = g\fr{\Tilde{X}, \Tilde{Z}}$, this is the ``realistic'' MLE from the estimated latent postions, $\fr{\Tilde{X}, \Tilde{Z}}$.
        \item $B^*_v = g\fr{X^*, Z^*}$, this is some MLE from some arbitary latent positions $(X^*, Z^*)$ near $(X, Z)$.
    \end{enumerate}

     Now we proceed to show that the MLE, $\Tilde{B}$, computed using the approximations, $\Tilde{X}, \Tilde{Z}$ gets sufficiently close to the true MLE, $\widehat{B}$ with $n$ large enough. Define $\Lambda(\epsilon) = \set{i \in V}{Z_{i*} \in D_p(\epsilon)}$ to be the set of node embeddings that are at least $\epsilon$ away from $0$ in all directions. Let $H^{*-1}_n = H^{-1}_n\fr{B^*; X^*, Z^*}$, from the mean value theorem, there is some $\fr{X^*, Z^*}$ on the line segment connecting $\fr{\Tilde{X}, \Tilde{Z}}$ and $\fr{X, Z}$, such that:
    
    \begin{align*}
        &\norm{\Tilde{B} - \widehat{B}}_2\\
        =& \norm{\left.\pderif{\fr{R, S}}{g\fr{R, S}}\right|_{(R,S) = (X^*, Z^*)}\fb{\fr{\Tilde{X}, \Tilde{Z}} - \fr{X, Z}}}_2\\
        =& \norm{ H^{*-1}_n\left.\pderif{(R_v, S_v)}{s_n\fr{\Theta_v; R, S}^T}\right|_{(\Theta_v, R, S) = (B^*_v, X^*, Z^*)} \fr{\fr{\Tilde{X}_v - X_v}, \fr{\Tilde{Z}_v - Z_v}}}_2.\\
    \end{align*}
    \noindent Here we use Taylor's Theorem to bound the distance of $\Tilde{B}$ and $\widehat{B}$ because $g$ is a continuously differentiable function. Next we split the matrix operation into row-wise operation:
    \begin{align*}
        =& \norm{H_n^{*-1} \fb{\left. \pderif{R_v}{s_n\fr{\Theta_v; R, Z^*}}^T\right|_{R = X^*}\fr{\Tilde{X}_v - X_v} + \left.\pderif{S_v}{s_n\fr{\Theta_v;X^*, S}}^T\right|_{S = Z^*} \fr{\Tilde{Z}_v - Z_v}}}_2\\
        =& \norm{H_n^{*-1} \s{i = 1}{n} \fc{\left.\pderif{R_{i*}}{s_n\fr{\Theta; R, Z^*}}^T\right|_{R = X^*} \fr{\Tilde{X}_{i*} - X_{i*}} + \left.\pderif{S_{i*}}{s_n\fr{\Theta_v; X^*, S}}^T\right|_{S = Z^*}\fr{\Tilde{Z}_{i*} - Z_{i*}}}}_2.\\
    \end{align*}
    \noindent Next we bound the norm above. Let $\frac{C_1}{n}$ be an upperbound for $\norm{H_n^{*-1}}_2$:
    \begin{align*}
        \leq& \frac{C_1}{n}\fr{\norm{\s{i = 1}{n} \left.\pderif{R_{i*}}{s_n\fr{\Theta_v;R, Z^*}}^T\right|_{R = X^*}}_2\norm{\Tilde{X} - X}_{2\to \infty} + \norm{\s{i = 1}{n} \left.\pderif{S_{i*}}{s_n\fr{\Theta_v; X^*, S}}^T\right|_{S = Z^*} }_2 \norm{\Tilde{Z} - Z}_{2\to \infty}}\\
        \leq& C_1\frac{\epsilon}{n} \fc{\norm{\s{i = 1}{n} \left.\pderif{R_{i*}}{s_n\fr{\Theta_v;R, Z^*}}\right|_{R = X^*}}_2 + \norm{\s{i = 1}{n} \left.\pderif{S_{i*}}{s_n\fr{\Theta_v; X^*, S}}\right|_{S = Z^*}}_2 }\\
        \leq& C_2 \frac{\epsilon}{n}\fc{\s{i = 1}{n}\s{j= 1}{p+1}\fr{ \abs{\log\fr{Z^*_{ij}}} + \frac{1}{Z^*_{ij}}}}.\\
        \end{align*}
        The bound above is given by Lemma \ref{lemma: score norm bound} where $C_2 \in \mathbb{H}$. Next we define $\xi_{ij} = \max\fc{Z_{ij} - \epsilon, \epsilon}$:
        \begin{align*}
        \leq& C_2 \frac{\epsilon}{n}\s{j = 1}{p+1}\fc{\s{i \in \Lambda(\epsilon)}{}\fr{\abs{\log\fr{\xi_{ij}}} + \frac{1}{\xi_{ij}}}
        +
        \s{i \in V - \Lambda(\epsilon)}{}\fr{\abs{\log\fr{Z_{ij}}} + \frac{1}{Z_{ij}}}}\\
        \leq& C_2 \frac{\epsilon}{n}\s{j = 1}{p+1} \s{i = 1}{n} \abs{\log\fr{\frac{Z_{ij}}{2}}} + \frac{2}{Z_{ij}} \quad \text{since $\xi_{ij} > \frac{Z_{ij}}{2}$ when $i \in \Lambda(\epsilon)$}.
    \end{align*}
    
\noindent For a fixed $j$, conditioning on $X$, $Z_{ij}$ are independent Beta random variables with distribution given by $\text{Beta}\fr{\alpha_{ij}, \sum\limits_{k \neq j}\alpha_{ik}}$. By assumption, $\alpha_{ij} > 2 + C_0$ for some fixed $C_0 \in \R^+$. So $\log\fr{Z_{ij}}$ has uniformly bounded first and second moments, and the some thing holds for $Z_{ij}^{-1}$ by Lemma \ref{lemma: moments of Beta inverse distr}. Let $\zeta_{ij} = \abs{\log\fr{\frac{Z_{ij}}{2}}} + \frac{2}{Z_{ij}}$, and let $\mu_{ij}, \sigma^2_{ij}$ be the mean and variance of $\zeta_{ij}$ respectively, then by Chebyshev's inequality\cite{casella2002statistical}, for any $\delta > 0$:
\begin{align*}
    P\fr{\frac{1}{n}\abs{\s{i = 1}{n}\zeta_{ij} - \mu_{ij}} > \delta} = \E\fr{P\fr{\given{\frac{1}{n}\abs{\s{i = 1}{n}\zeta_{ij} - \mu_{ij}} > \delta}X}} < \frac{1}{n\delta^2} \max_{i \in V}\sigma^2_{ij}.
\end{align*}

\noindent Choose any constant $\delta$, then $\sum\limits_{i = 1}^{n} \zeta_{ij} = O_p(n)$, and with high probability: $\norm{\Tilde{B} -\widehat{B}}_2 \leq C_3 \epsilon$. Here $C_3$ depends on $C_2$, $p$, $\delta$, $\sum\limits_{i = 1}^{n}\mu_{ij}$. 
\end{proof}

\section{Supporting Lemmas}
\subsection{Lemmas for Theorem \ref{theorem: suff cond for consistency}}
Recall the following definitions from previous sections:
\begin{alignat*}{3}
    &N_i 
    = \sum_{j \in \tau_w(i)} Z_{j}Y_{ij}, \quad 
    &&N^*_i 
    = \E\fr{\left.\sum_{j \in \tau_w(i)} Z_{j}Y_{ij} \right| Z }, \quad 
    &&\widehat{N}_i 
    = \E\fr{\left.\sum_{j \in \tau_w(i)} Z_{j}Y_{ij} \right| Z_i }\\
    &D_i 
    = \sum_{j \in \tau_w(i)} Y_{ij}, \quad 
    &&D^*_i 
    = \E\fr{\left.\sum_{j \in \tau_w(i)} Y_{ij} \right| Z }, \quad 
    &&\widehat{D}_i 
    = \E\fr{\left.\sum_{j \in \tau_w(i)} Y_{ij} \right| Z_i }\\
    &A^w_i 
    = N_i D^{-1}_i, \quad
    &&A^{w*}_i 
    = N^*_i D^{*-1}_i, \quad
    &&\widehat{A}^w_i 
    = \widehat{N}_i\widehat{D}^{-1}_i
\end{alignat*}
\noindent To prove theorem \ref{theorem: suff cond for consistency}, we essentially need to argue that for nodes in $\Lambda_g$ (with decent connectivity), our estimate $\widehat{X}$ is very close to $X$ (Lemma \ref{lemma: X_lambda_g bound}). While we can't say the same about nodes with bad connectivity, we show that under the assumptions of theorem \ref{theorem: suff cond for consistency}, there will be so few nodes with bad connectivity that they don't matter (Lemma \ref{lemma: size of Lambda b}). As for Lemma \ref{lemma: N - N_star}, \ref{lemma: Nstar - Nhat},  \ref{lemma: A - Ahat}, they are a combination of union bounds and Bernstein-type bound \cite{vershynin2018high} that lead to $\ref{lemma: X_lambda_g bound}$.

\begin{lemma} \label{lemma: N - N_star}
For all $\lambda > 0$:
\begin{align*}
    P\fr{\norm{N_i - N_i^*}_2 \geq \lambda n} &\leq 2p \exp\fc{-\frac{2\lambda^2n}{p}}\\
    P\fr{\abs{D_i - D_i^*} \geq \lambda n} &\leq 2\exp\fc{-2\lambda^2 n}
\end{align*}
\end{lemma}

\begin{proof}
\begin{align*}
    &P\fr{\frac{1}{n}\norm{N_i - N_i^*}_2 \geq \lambda}\\
    \leq &P\fr{\frac{1}{n}\norm{N_i - N^*_i}_\infty \geq \frac{\lambda}{\sqrt{p}}}\\
    =& \E\fr{P\fr{\left. \frac{1}{n}\norm{N_i - N_i^*}_\infty \geq \frac{\lambda}{\sqrt{p}} \right |Z}} \\
    =& \E\fr{P \fr{\left. \frac{1}{n} \norm{\sum_{j \in \tau(i)} Y_{ij}Z_{j} - \sum_{j \in \tau(i)} \E(Y_{ij})Z_{j}}_\infty \geq \frac{\lambda}{\sqrt{p}} \right| Z}}\\
    =& \E \fr{P \fr{ \left.\bigcup_{l = 1}^p \fc{ \frac{1}{n} \abs{\sum_{j \in \tau(i)} Y_{ij}Z_{jl} - \sum_{j \in \tau(i)} \E(Y_{ij})Z_{jl}} > \frac{\lambda}{\sqrt{p}} }\right| Z}}\\
    \leq& \s{l = 1}{p} \E \fr{P\fr{\left. \frac{1}{n} \abs{\sum_{j \in \tau(i)} Y_{ij}Z_{jl} - \sum_{j \in \tau(i)} \E(Y_{ij})Z_{jl}} > \frac{\lambda}{\sqrt{p}} \right| Z}}\\
    \leq& 2p \exp\fc{-\frac{2\lambda^2n}{p}},
\end{align*}
by Hoeffding's' Inequality, since $Y_{ij}$ are independent r.v. when conditioning on $Z$.

\begin{align*}
    &P\fr{\frac{1}{n}\abs{D_i - \E\fr{D_i}} \geq \lambda}\\
    &= \E\fr{P\fr{\left. \frac{1}{n}\abs{D_i - \E\fr{D_i}} \geq \lambda \right| Z} }\\
    &= \E \fr{P\fr{\left. \abs{\frac{1}{n}\sum_{j \in \pi(i)} Y_{ij} - \frac{1}{n}\sum_{j \in \pi(i)} \E(Y_{ij}) } \geq \lambda \right| Z}}\\
    &\leq 2\exp\fc{-2\lambda^2 n}.
\end{align*} 
\end{proof}

\begin{lemma} \label{lemma: Nstar - Nhat} 
Let $c = \frac{\abs{\tau_w(i)}}{n}$, then for all $\lambda > 0$:
    \begin{align*}
        P\fr{ \norm{N_i^* - \widehat{N}_i}_2 \geq n\lambda}
        \leq& 2p \exp \fr{\frac{-3n\lambda^2}{12c + 4\lambda}}\\
        P\fr{\norm{D^*_i - \widehat{D}_i}_2 \geq n\lambda}
        \leq& (p+1) \exp \fr{\frac{-3n\lambda^2}{12c + 4\lambda}}
    \end{align*}
\end{lemma}

\begin{proof}
We bound $P\fr{\norm{N^*_i - \widehat{N}_i}_2 \geq n\lambda}$, and $P\fr{\abs{D^*_i - \widehat{D}_i} \geq  n\lambda}$ separately using the Bernstein inequality. First we give an upper bound for $\norm{\widehat{N}_i - N_i^*}_2$:

\begin{align*}
    \norm{N_i^* - \widehat{N}_i}_2 &= \norm{\E\fr{\left.\sum_{j \in \tau_w(i)} Z_{j}Y_{ij} \right| Z } - \E\fr{\left.\sum_{j \in \tau_w(i)} Z_{j}Y_{ij} \right| Z_i } }_2\\
    &= \norm{ \sum_{j\in \tau_w(i)} Z_j \E\fr{Y_{ij} | Z} - \sum_{j\in \tau_w(i)} \E\fr{\E\fr{Z_jY_{ij} | Z_i, Z_j}} }_2\\
    &= \norm{\sum_{j\in \tau_w(i)} Z_j Z_j^TZ_i - \E(Z_j Z_j^T)Z_i}_2\\
    & \leq \norm{\fr{\sum_{j\in \tau_w(i)} Z_jZ_j^T} - \E\fr{\sum_{j\in \tau_w(i)} Z_jZ_j^T}}_2 \norm{Z_i}_2\\
    &\leq \norm{\fr{\sum_{j\in \tau_w(i)} Z_jZ_j^T} - \E\fr{\sum_{j\in \tau_w(i)} Z_jZ_j^T}}_2.
\end{align*}

\noindent Similarly, for $\norm{D^*_i - \widehat{D}_i}_2$:
\begin{align*}
    \norm{D^*_i - \widehat{D}_i}_2 = \norm{\E\fr{\left.\sum_{j \in \tau(i)} Y_{ij} \right| Z} - \E\fr{\left. \sum_{ j \in \tau(i)}  Y_{ij} \right| Z_i}}_2 \leq \norm{\sum_{j \in \tau(i)} Z_j - \E\fr{Z_j}}_2.
\end{align*}

\noindent Apply the matrix Bernstein inequality to $\frac{1}{n} \norm{\fr{\sum_{j\in \tau_w(i)} Z_jZ_j^T - \E\fr{ Z_jZ_j^T}}}_2$ and $\frac{1}{n}\norm{\sum_{j \in \tau(i)} Z_j - \E\fr{Z_j}}_2$, we get the following lower bounds:
\begin{align*}
    P\fr{ \norm{N_i^* - \widehat{N}_i}_2 \geq n\lambda}
    \leq& 2p \exp \fr{\frac{-3\lambda^2}{6v_N + 2\lambda L_N}}
    = 2p \exp \fr{\frac{-3n\lambda^2}{12c + 4\lambda}}\text{, where $c = \frac{\abs{\tau_w(i)}}{n}$}\\
    P\fr{\norm{D^*_i - \widehat{D}_i}_2 \geq n\lambda}
    \leq& (p+1)\exp\fc{  \frac{-3\lambda^2}{6v_D + 2\lambda L_D}} = (p+1) \exp \fr{\frac{-3n\lambda^2}{12c + 4\lambda}}.
\end{align*}

\noindent $L_N,\ L_D,\ v_N,\ v_D$ are constants defined as below:
\begin{align*}
    L_N &\geq \frac{1}{n} \norm{Z_jZ_j^T - \E\fr{Z_jZ_j^T}}_2,\\
    L_D &\geq \frac{1}{n} \norm{Z_i - \E\fr{Z_i}}_2,\\
    v_N &\geq V\fr{\frac{1}{n}\sum_{j \in \tau_w(i)} Z_jZ_j^T} = \frac{1}{n^2}\norm{\E\fr{\fb{\sum_{j \in \tau_w(i)} Z_jZ_j^T - \E\fr{Z_jZ_j^T}}^2}}_2,\\
    v_D &\geq V\fr{\frac{1}{n}\sum_{j \in \tau_w(i)} Z_j} = \frac{1}{n^2}\norm{\E\fr{\fb{\sum_{j \in \tau_w(i)} Z_j - \E\fr{ Z_j}}^T\fb{\sum_{j \in \tau_w(i)} Z_j - \E\fr{ Z_j}}}}_2.\\
\end{align*}

\noindent We get $L_N = L_D = \frac{2}{n},\ v_N = v_D = \frac{2c}{n}$ from the computation below: 
\begin{alignat*}{3}
    &L_N:\ \frac{1}{n}\norm{Z_jZ_j^T - \E\fr{Z_jZ_j^T}}_2  
    &&\leq \frac{1}{n} \fr{\norm{Z_jZ_j^T}_F + \norm{\E\fr{Z_jZ_j^T}}_F} 
    &&\leq \frac{2}{n},\\
    &L_D:\ \frac{1}{n}\norm{Z_i - \E\fr{Z_i}}_2
    &&\leq \frac{1}{n}\fr{\norm{Z_i}_2 + \norm{\E\fr{Z_i}}_2}
    &&\leq \frac{2}{n},\\
\end{alignat*}

\begin{align*}
    v_N:\ V\fr{\frac{1}{n}\sum_{j \in \tau_w(i)} Z_jZ_j^T} =& \frac{1}{n^2} \norm{\E\fr{\fb{\sum_{j \in \tau_w(i)} Z_jZ_j^T - \E\fr{ Z_jZ_j^T}}^2}}_2\\ 
    =& \frac{1}{n^2} \norm{\sum_{j \in \tau_w(i)}\E\fr{ \fr{Z_jZ_j^T - \E\fr{Z_jZ_j^T}}^2}}_2 \quad \text{since $Z_i, Z_j$ are independent for $i \neq j$}\\
    \leq&\frac{1}{n^2} \sum_{j \in \tau_w(i)} \norm{\E\fr{ \fr{Z_jZ_j^T - \E\fr{Z_jZ_j^T}}^2}}_2\\
    \leq&\frac{1}{n^2} \sum_{j \in \tau_w(i)}  2\norm{\E\fr{ Z_jZ_j^T}^2}_F\\
    \leq& \frac{2c}{n} \quad \text{where $c = \frac{\abs{\tau_w(i)}}{n}$},
\end{align*}

\begin{align*}
    v_D:\ V\fr{\frac{1}{n} \sum_{j \in \tau_w(i)} Z_j} &= \frac{1}{n^2}\norm{\E\fr{\fb{\sum_{j \in \tau_w(i)} Z_j - \E\fr{ Z_j}}\fb{\sum_{j \in \tau_w(i)} Z_j - \E\fr{ Z_j}}^T}}_2\\
    &= \frac{1}{n^2}\norm{\sum_{j \in \tau_w(i)}\E\fr{ \fr{Z_j - \E\fr{Z_j}}\fr{Z_j - \E\fr{Z_j}}^T }}_2 \\
    &\leq \frac{1}{n^2} \sum_{j \in \tau_w(i)} 2\norm{\E\fr{Z_jZ_j^T }}_F\\
    &\leq  \frac{2c}{n}.
\end{align*}
\end{proof}
\begin{lemma} \label{lemma: A - Ahat}
    Let $\sigma, \lambda \in (0, 1)$, $\forall i \in \Lambda_g$. Let $c_n= 8p \exp\fc{-\frac{\lambda \sigma^2 n}{80p}}$. It holds that
    \begin{align*}
        P\fr{\norm{A^w_i - \widehat{A}^w_i}_2^2 < \lambda} \geq 1 - c_n\text{ and }P\fr{\norm{A^b_i - \widehat{A}^b_i}_2^2 < \lambda} \geq 1- c_n.
    \end{align*}
\end{lemma}

\begin{proof} 
    We will prove the case for $A^w$ here, and the exact same arguments will work for $A^b$. For $y > 0, i \in \Lambda_g$, if $\norm{N_i - N^*_i}_2 \leq y \text{ and }|D_i - D^*_i| \leq y$, then

\begin{align*}
    \norm{A_i^w - A_i^{w*}}_2
    =& \norm{\frac{N_i - N_i^*}{D^*_i}+ 
    \frac{N_i\fr{D_i^* - D_i}}{D_iD_i^*}}_2\\
    \leq& \frac{\norm{N_i - N^*_i}_2}{D^*_i} + 
    \frac{\norm{N_i}_2|D_i - D^*_i|}{D_iD^*_i}\\
    \leq& \frac{y}{D^*_i} + \frac{(\norm{N^*_i}_2+ y)y}{(D^*_i - y)D^*_i} \\
    =& \frac{y\fr{D^*_i + \norm{N_i^*}_2}}{D^*_i\fr{D_i^* - y}}.
\end{align*}
\noindent If $y = \frac{\sqrt{\lambda} D^{*2}_i}{D^*_i(1+ \sqrt{\lambda}) + \norm{N_i^*}_2}$, then $\norm{A_i^w - A_i^{w*}}_2 \leq \sqrt{\lambda}$. \\

For all $i \in \Lambda_g$, $D^*_i > \sqrt{\sigma }n$, so $\frac{\sqrt{\lambda} D^{*2}_i}{D^*_i(1+\sqrt{\lambda}) + \norm{N_i^*}_2} \geq \frac{\sqrt{\lambda} \sigma n^2}{n(1 + \sqrt{\lambda}) + n} \geq \frac{\sqrt{\lambda} \sigma  n}{3}$ (since $D^*_i, \norm{N^*_i}_2 < n$), and:
\begin{align*}
    P\fr{\norm{A_i^w - A_i^{w*}}_2 < \sqrt{\lambda}} >& P\fr{\norm{N_i - N^*_i}_2 < \frac{\sqrt{\lambda} \sigma  n}{3} 
    \text{ and } \abs{D_i - D^*_i} < \frac{\sqrt{\lambda} \sigma  n}{3}}\\
    >& 1 - P\fr{\norm{N_i - N^*_i}_2 \geq \frac{\sqrt{\lambda} \sigma  n}{3} } - P \fr{\abs{D_i - D^*_i} \geq \frac{\sqrt{\lambda} \sigma  n}{3}}\\
    >& 1 - 2p \exp\fc{ -\frac{2 \lambda \sigma^2 n}{9p}} - 2 \exp\fc{-\frac{2\lambda \sigma^2 n}{9}}\\
    >& 1 - 4p \exp\fc{ -\frac{\lambda \sigma^2 n}{20p}}.
\end{align*}
\noindent Similarly:
\begin{align*}
    P\fr{\norm{\widehat{A}_i^w - A_i^{w*}}_2 < \sqrt{\lambda}} >& P\fr{\norm{\widehat{N}_i - N^*_i}_2 < \frac{\sqrt{\lambda} \sigma  n}{3} 
    \text{ and } \abs{\widehat{D}_i - D^*_i} < \frac{\sqrt{\lambda} \sigma  n}{3}}\\
    >& 1 - P\fr{\norm{\widehat{N}_i - N^*_i}_2 \geq \frac{\sqrt{\lambda} \sigma  n}{3} } - P \fr{\abs{\widehat{D}_i - D^*_i} \geq \frac{\sqrt{\lambda} \sigma  n}{3}}\\
    >& 1 - 2p \exp\fc{ -\frac{\lambda \sigma^2 n}{36c + 4 \sigma \sqrt{\lambda}}} - (p+1) \exp\fc{ -\frac{\lambda \sigma^2 n}{36c + 4 \sigma \sqrt{\lambda}}}  \\
    >& 1 - 4p \exp\fc{ -\frac{\lambda \sigma^2 n}{40}} \quad \text{recall that $c = \frac{\abs{\tau_w(i)}}{n}$, and $c, \sigma, \lambda \leq 1$}\\
    >& 1 - 4p \exp\fc{ -\frac{\lambda \sigma^2 n}{20p}} \quad \text{since $p \geq 2$}
\end{align*}
\noindent Comebine the two upperbounds, we get:
\begin{align*}
    P\fr{\norm{A^w_i - \widehat{A}^w_i}_2 < \sqrt{\lambda}} \geq& P\fr{\norm{A_i^w - A_i^{w*}}_2 + \norm{\widehat{A}_i^w - A_i^{w*}}_2 < \sqrt{\lambda}}\\
    \geq& 1 - P\fr{\norm{A_i^w - A_i^{w*}}_2 < \frac{\sqrt{\lambda}}{2}} - P\fr{\norm{\widehat{A}_i^w - A_i^{w*}}_2 < \frac{\sqrt{\lambda}}{2}} \\
    \geq& 1 - 8p \exp\fc{ -\frac{\lambda \sigma^2 n}{80p}}.
\end{align*}
\end{proof}
\begin{lemma} \label{lemma: X_lambda_g bound} 
    For $X, \widehat{X}$ defined in section \ref{def: X_star}, $\lambda, \sigma \in (0, 1)$, we have:
    \begin{align*}
        P\fr{\norm{X_{\Lambda_g} - \widehat{X}_{\Lambda_g}}_{2 \to \infty}^2 < \lambda} \geq 1 - 16pn\fr{\exp\fc{ -\frac{\lambda \sigma^2 n}{160p} }}.
    \end{align*}
\end{lemma}
\begin{proof}
    \begin{align*}
    & P\fr{\norm{X_{\Lambda_g} - \widehat{X}_{\Lambda_g}}_{2 \to \infty}^2 < \lambda}\\
    =& P\fr{ \max_{i \in \Lambda_g} \norm{X_i - \widehat{X}_i}_2^2 < \lambda}\\
    =& P\fr{ \bigcap_{i \in \Lambda_g}\fc{\norm{A_i^w - \widehat{A}_i^w}_2^2 + \norm{A_i^b - \widehat{A}_i^b}_2^2 < \lambda} }\\
    \geq& 1 - \s{i\in \Lambda_g}{} \fr{P\fr{\norm{A_i^w - \widehat{A}_i^w}_2^2 \geq \frac{\lambda}{2}} + P\fr{\norm{A_i^b - \widehat{A}_i^b}_2^2 \geq \frac{\lambda}{2}}}\\
    \geq& 1 - 16pn\fr{\exp\fc{ -\frac{\lambda \sigma^2 n}{160p} }}.
\end{align*}
\end{proof}
\begin{lemma} \label{lemma: size of Lambda b} 
    Let $0 < \sigma < 1$. Assume, $\frac{1}{n}D^*_i$ has a density, $f$, such that $f(x) \leq k_bx^{-\delta_b}$ for some $\delta_b \in (0,1), k_b > 0$ on $(0, 2\sqrt{\sigma})$, then the following holds true:
    \begin{align*}
        P\fr{\abs{\Lambda_b} \leq n \fr{\sqrt{\sigma} + \frac{2k_b}{1 - \delta_b}\sigma^{\frac{1 - \delta_b}{2}}}} > 1 - (2pn+1)\exp\fc{-\frac{\sigma n }{19}}.
    \end{align*}
    
    \noindent Note that to have $\abs{\Lambda_b} = o_p(n)$ and $\norm{X_{\Lambda_g} - \widehat{X}_{\Lambda_g}}_{2 \to \infty}^2 = o_p(1)$ at the same time, we need $\sigma \in \omega(n^{-\frac{1}{2}}) \cap o(1)$.
\end{lemma}
\begin{proof}
    Define $W_i = \mathbf{1}_{\fc{\widehat{D}_i \leq \frac{3}{2}\sqrt{\sigma}n}}$. Note that $\abs{\Lambda_b} = \mathbf{1}_{\fc{D^*_i \leq \sqrt{\sigma}n}}$, and if the the following conditions hold:
    \begin{enumerate}
        \item $\abs{D^*_i - \widehat{D}_i} \leq \frac{1}{2}\sqrt{\sigma}n$ for all $i \in V$,
        \item $\abs{\sum\limits_{i \in V}W_i - \E\fr{W_i}} \leq \sqrt{\sigma}n$,
    \end{enumerate}
    then the following inequalities hold:
    \begin{align*}
        \abs{\Lambda_b} \leq \s{i \in V}{}W_i \leq  \s{i \in V}{}\E\fr{W_i} + \sqrt{\sigma}n \leq \s{i\in V}{}\E\fr{\mathbf{1}_{\fc{D^*_i \leq 2\sqrt{\sigma}n}}} + \sqrt{\sigma}n .
    \end{align*}
    \noindent By assumption, we find the following upper bound:
    \begin{align*}
        \E\fr{\mathbf{1}_{\fc{D^*_i \leq 2\sqrt{\sigma}n}}} = P\fr{D^*_i \leq 2\sqrt{\sigma}n}
        \leq \int_0^{2\sqrt{\sigma}} k_b x^{-\delta_b}dx
        = \frac{k_b}{1 - \delta_b}(2\sqrt{\sigma})^{1 - \delta_b} \leq \frac{2k_b}{1 - \delta_b}\sigma^{\frac{1 - \delta_b}{2}}.
    \end{align*}
    \noindent Combine everything above, we have:
    \begin{align*}
        P\fr{\abs{\Lambda_b} \leq n \fr{\sqrt{\sigma} + \frac{2k_b}{1 - \delta_b}\sigma^{\frac{1 - \delta_b}{2}}}} \geq& P\fr{\bigcap_{i \in V}\fc{\abs{D^*_i - \widehat{D}_i} \leq \frac{1}{2}\sqrt{\sigma}n} \cap \fc{\abs{\sum\limits_{i \in V}W_i - \E\fr{W_i}} \leq \sqrt{\sigma}n}}\\
        \geq& 1 - \s{i \in V}{} P\fr{\abs{D_i^* - \widehat{D}_i} > \frac{1}{2}\sqrt{\sigma}n} - P\fr{\abs{\sum\limits_{i \in V}W_i - \E\fr{W_i}} \leq \sqrt{\sigma}n}\\
        \geq& 1 - n(p + 1) \exp\fc{-\frac{3\frac{\sigma}{4}n}{(12 + 4\frac{\sqrt{\sigma}}{2})}} - \exp\fc{-2 \sigma n}\\
        \geq& 1 - (2pn+1)\exp\fc{-\frac{\sigma n }{19}}.
    \end{align*}
    
\end{proof}

\subsection{Lemmas for Theorem \ref{theorem: oracle align consistency}}
Theorem \ref{theorem: oracle align consistency} is mainly about applying the implicit function theorem to bound the perturbation of MLE caused by having to "estimate" data. The problem is that the function, $g$, that maps data to MLE diverges near 0. So we need to shave off the portion of our data that is near 0. Lemma \ref{lemma: latent position concentration} guarantees that after deleting data, we still have enough left for inference, and Lemma \ref{lemma: score norm bound}, \ref{lemma: moments of Beta inverse distr} helps us characterize the function $g$. Lemma \ref{lemma: ASE premise} is about showing that under our assumptions, ASE is consistent, which means that we can use ASE as an estimate of our data. 

\begin{lemma} \label{lemma: latent position concentration} 
    Let $Z_{i, t}$ be defined at Table \ref{table: definitions}. For all $A \subset \Delta^p$ with a positive Lebesgue measure:
    \begin{enumerate}
        \item     $\s{i = 1}{n}\ind_{\fc{Z_{i, 0} \in A}} = \Theta_{P}(n)$,
        \item $\s{i = 1}{n}\ind_{\fc{Z_{i, t} \in A}} = \Theta_P(n) \implies \s{i = 1}{n}\ind_{\fc{Z_{i, t+1} \in A}} = \Theta_{P}(n)$.
    \end{enumerate}
\end{lemma}
\begin{proof} 
    At $t = 0$, by assumption $Z_{i, 0}$ are non-degenerate i.i.d. Dirichlet random variables for $i = 1, ..., n$. Let $\mu$ be the Lebesgue measure for $\R^p$. For all $A \subset \Delta^p$ with $\mu(A) > 0$, $\exists \delta > 0$ such that $\forall x \in A$, $f_{Z_{i, 0}}(x) > \delta$. Therefore $P\fr{Z_{i, 0} \in A} > \delta \mu(A)$, and:
    \begin{align*}
        \E\fr{\s{i = 1}{n} \ind_{Z_{i,0} \in A}} = \s{i = 1}{n}P\fr{Z_{i, 0} \in A} > n \delta \mu(A) = \Theta(n).    
    \end{align*}
    Since $Z_{i, 0}$ are i.i.d., $U_{i,0} = \ind_{Z_{i, 0} \in A}$ are i.i.d. Bernoulli random variables. Through Hoefdding's inequality\cite{vershynin2018high}, we have:
    \begin{align*}
        P\fr{\abs{ \s{i = 1}{n} U_{i,0} - \E\fr{U_{i,0}}} > \epsilon} \leq 2\exp\fc{-\frac{2\epsilon^2}{n}}.
    \end{align*}
    \noindent Take $\epsilon \in \omega(\sqrt{n}) \cap o(n)$, and we have:
    \begin{align*}
        \s{i = 1}{n} \ind_{Z_{i, 0} \in A} = \Theta_p(n)
    \end{align*}
    as desired.\\

     Now assume $\s{i = 1}{n} \ind_{Z_{i,t} \in A} = \Theta_p(n)$. By definition, $Z_{i, t+1} \sim \Dir\fr{\alpha_{i, t+1}}$ where $\alpha_{i, t+1} = \exp \fc{X_{i, t}^T B}$. Since $X_{i,t}$ are uniformly bounded, $\alpha_{i, t+1}$ are positive and uniformly bounded for $i = 1,...,n$. Similar to $t = 0$,  $\exists \delta > 0$ s.t. $\forall  x \in A$, $f_{Z_{i, t+1}}(x) > \delta$ for $i = 1,..., n$. So $P\fr{Z_{i, t+1} \in A | \alpha_{t+1}} > \delta \mu(A)$. Given $\alpha_{t+1}$, $Z_{i, t+1}$ are independent. So:
    \begin{align*}
        \E\fr{\s{i = 1}{n} \ind_{Z_{i,t+1} \in A}} &= \E\fr{\E\fr{\left. \s{i = 1}{n} \ind_{Z_{i,t+1} \in A} \right| \alpha_{t+1}}}\\
        &=\E\fr{\s{i = 0}{n} P\fr{\left. Z_{i, t+1} \in A \right|\alpha_{t+1}}}\\
        &\geq n \delta \mu(A).
    \end{align*}
    Then using Hoeffding's inequality, we have:
    \begin{align*}
        P\fr{\abs{ \s{i = 1}{n} U_{i,t+1} - \E\fr{U_{i,t+1}}} > \epsilon} &= \E\fr{P\fr{\abs{ \s{i = 1}{n} U_{i,t+1} - \E\fr{U_{i,t+1}}} > \epsilon| \alpha_{t+1}}}\\
        &\leq 2\exp\fc{-\frac{2\epsilon^2}{n}}.
    \end{align*}
    Again, take $\epsilon \in \omega(\sqrt{n}) \cap o(n)$, and we have the desired result.
\end{proof}
\begin{lemma} \label{lemma: ASE premise}
     The following conditions hold for $Z_t$ for $t = 0, 1$:
    \begin{enumerate}
        \item $\lambda_p\fr{Z_tZ_t^T} = \Theta_{p}(n)$, where $\lambda_p(A) =$ the $p^{th}$ largest singular value of $A$,
        \item $\delta\fr{Z_tZ_t^T} = \Theta_{p}(n)$, where $\delta(P) = \max_i \sum_j P_{ij}$.
    \end{enumerate}
     If the above conditions holds, then for $\widehat{Z}_t$, the ASE-estimate of $Z_t$:
    \begin{align*}
        \min_{W \in O_p}\norm{Z_t -  \widehat{Z}_tW}_{2 \to \infty} \leq \frac{C\log^2(n)}{\delta^{1/2}\fr{Z_tZ_t^T}}.
    \end{align*}
\end{lemma}

\begin{proof}
    First we prove that $\lambda_p\fr{Z_tZ_t^T} = \Theta_p(n)$:\\
    Let $b_1,..., b_p$ be a basis of $\Delta^p$. Let $A_k$ be an open neighborhood of $b_k$ for $k = 1,..., p$, such that $A_i$ and $A_j$ are disjoint for any $i \neq j$. It suffices to show that:
    \begin{align*}
        u^T \fr{\s{i = 1}{n} Z_{i, t}Z_{i, t}^T} u > cn \text{ for all non-zero $u \in \R^p$ w.h.p.}.
    \end{align*}
    \noindent Fix $u$, then $\exists \delta > 0, k \in \fc{1,...,p}$ such that  $\forall x \in A_k$, $u^T x > \delta$. By lemma \ref{lemma: latent position concentration}, the number of $Z_{i,t}$ in each $A_k$ is $\Theta_p(n)$. Therefore:
    \begin{align*}
        u^T \fr{\s{i = 1}{n} Z_{i, t}Z_{i, t}^T} u = \s{i = 1}{n} \norm{Z^T_{i,t}u}^2_2 \geq \sum_{i:Z_{i, t} \in A_k} \norm{Z^T_{i,t}u}^2_2 = \Theta_p(n).
    \end{align*}

    \noindent Next we prove that $\delta\fr{Z_tZ_t^T} = \Theta_p(n)$:\\
     Pick any $A_k$, WLOG, assume $Z_{1, t} \in A_k$, then there exists $\epsilon > 0$ such that $Z_{1,t}^T Z_{j, t} > \epsilon$ for all $Z_{j, t} \in A_k$. Therefore we have:
    \begin{align*}
        \delta\fr{Z_tZ_t^T} = \max_{i \leq n} Z_{i, t}^T \s{j = 1}{n} Z_{j, t} \geq Z_{1, t}^T \sum_{j: Z_{j, t} \in A_k} Z_{j, t} > \epsilon c n \quad \text{for some $c > 0$ independent of n}
    \end{align*}
    as desired.
     The last statement is Theorem 26 in \cite{RDPG_Survey}
\end{proof}

\begin{lemma} \label{lemma: score norm bound}
    Recall our score function $s_n$ is given by:
    \begin{align*}
        s_n(\Theta; R, S) =& \s{i = 1}{n} \fr{R_{i*} \otimes I_{p+1}} \diag \fr{\alpha_i(\Theta, R_{i*})}\fr{\log(S_{i*}) - \mu_i(\Theta, R_{i*})}\\
        \text{where } \alpha_i\fr{\Theta, R_{i*}} =& \exp\fc{R_{i*}^T\Theta} \text{, and }\mu_i\fr{\Theta, R_{i*}} = \psi\fr{\alpha_i} - \psi\fr{\s{j = 1}{p+1}\alpha_{ij}}.
    \end{align*}
    We have the following bounds on the 2-norm of the partial derivatives of $s_n$ evaluated at some point $X^* \in \R^{}, Z^* \in \R^{}$ near the true design matrix $X$, and response matrix $Z$:
    \begin{align*}
    \norm{\left.\pderif{R_{i*}}{s_n\fr{\Theta_v;R, Z^*}}\right|_{R = X^*}}_2 \leq& C_2 \s{j = 1}{p+1}\abs{\log\fr{Z^*_{ij}}} \text{ for some $C_2 \in \R^+$ independent from $i$},\\
    \norm{\given{\pderif{S_{i*}}{s_n(\Theta_v; X^*, S)}}_{S = Z^*}}_2 \leq& C_3 \s{j = 1}{p+1}\frac{1}{Z_{ij}^*}\text{ for some $C_3 \in \R^+$ independent from $i$}.
\end{align*}
\end{lemma}

\begin{proof}
     Let $K = p+1$, $q = 3p+1$, and $\mathbf{K}_{m, n}$ be the $nm \times mn$ commutation matrix. Since all components of $X^*_{i*}$ are between $0$ and $1$, $\alpha_i\fr{\Theta, X^*_{i*}}$ is uniformly bounded from above and lower bounded away from $0$ for any fixed $\Theta$.
    
     For $\left.\pderif{R_{i*}}{s_n\fr{\Theta_v;R, Z^*}}\right|_{R = X^*}$:

    \begin{align*}
        &\pderif{R_{i*}}{s_n\fr{\Theta_v;R, Z^*}}\\
        =& \pderiv{R_{i*}}{\s{j = 1}{n} \fr{R_{j*} \otimes I_{K}} \diag \fr{\alpha_j(\Theta, R_{j*})}\fr{\log(Z^*_{j*}) - \mu_j(\Theta, R_{j*})}}\\
        =& \pderiv{R_{i*}}{\fr{R_{i*} \otimes I_{K}} \diag \fr{\alpha_i}\Delta_i}\text{, where $\Delta_i = \log(Z^*_{i*}) - \mu_i(\Theta, R_{i*})$, $\alpha_i = \alpha_i(\Theta, R_{i*})$}\\
        =& \fr{R_{i*}\otimes I_{K}}\pderiv{R_{i*}}{\Delta_i \circ \alpha_i} + \fr{I_{K^2q} \otimes \fb{\Delta_i \circ \alpha_i}}\pderiv{R_{i*}}{\fr{R_{i*} \otimes I_K}}.\\
    \end{align*}
\noindent With the applications of product rule and chain rule, we get:
    \begin{align*}
        \pderiv{R_{i*}}{\Delta_i \circ \alpha_i} &= \fr{\diag\fr{\Delta_i} - \Sigma_i\diag\fr{\alpha_i}} \diag\fr{\alpha_i}\Theta \text{, where $\Sigma_i = \diag\fr{\psi^{(1)}(\alpha_i)} - \psi^{(1)}\fr{\s{j= 1}{p+1}\alpha_{ij}}$},\\
        \pderiv{R_{i*}}{\fr{R_{i*} \otimes I_K}} &= \fr{I_{Kq} \otimes \mathbf{K}^T_{1,K} \otimes I_K} \fr{I_{Kq} \otimes \vecop(I_K)}.
    \end{align*}

\noindent When $R = X^*$, the terms above that may be infinite are $\log(Z^*_{i*}), \mu_i\fr{\Theta, X^*_{i*}}, \Sigma_i\fr{\Theta, X^*_{i*}}$. For $\mu, \Sigma$, the digamma function and the trigamma function $\psi, \psi^{(1)}$, are both monotone functions that diverges at $0$. Since all components of $\alpha_i$ is uniformly bounded away from $0$, the size of $\mu_i, \Sigma_i$ are uniformly bounded from above. There is no bound for $Z^*_{i*}$, so:
\begin{align*}
    \norm{\left.\pderif{R_{i*}}{s_n\fr{\Theta_v;R, Z^*}}\right|_{R = X^*}}_2 \leq C_2 \s{j = 1}{p+1}\abs{\log\fr{Z^*_{ij}}} \text{ for some $C_2 \in \R^+$ independent from $i$}.
\end{align*}

\noindent Next for $\given{\pderif{S_{i*}}{s_n\fr{\Theta_v; X^*, S}}}_{S = Z^*}$:
\begin{align*}
    &\norm{\given{\pderiv{S_{i*}}{s_n\fr{\Theta_v; X^*, S}}}_{S = Z^*}}_2\\
    =& \norm{\given{\pderiv{S_{i*}}{\s{j = 1}{n} \fr{X_{j*} \otimes I_{K}} \diag \fr{\alpha_j(\Theta, X_{j*})}\fr{\log(S^*_{j*}) - \mu_j(\Theta, X_{j*})}}}_{S = Z^*}}_2\\
    =& \norm{\given{\pderiv{S_{i*}}{\fr{X_{i*} \otimes I_{K}} \diag(\alpha_i) \log(S_{i*})}}_{S = Z^*}}_2\\
    \leq& C_3 \s{j = 1}{p+1}\frac{1}{Z_{ij}^*}\text{ for some $C_3 \in \R^+$ independent from $i$}.
\end{align*}
\end{proof}

\begin{lemma} \label{lemma: moments of Beta inverse distr} 
    Let $a, b \in \R^+$, consider a beta random variable\cite{dirichlet_definition}, $X \sim \text{Beta}(a, b)$. If $a > k$, then $\E\fr{X^{-k}} = \frac{\Gamma(a + b)\Gamma(a - k)}{\Gamma(a)\Gamma(a + b - k)} < \infty$.
\end{lemma}
\begin{proof}
    Let $B(x, y) = \int\limits_0^1 t^{x - 1} (1 - t)^{y - 1}dt$ be the Beta function for $x, y \in \R^+$. We shall compute $\E\fr{X^{-k}}$:
    \begin{align*}
        \E\fr{X^{-k}} =& B^{-1}(a, b)\int_0^1 x^{-k} x^{a - 1}(1- x)^{b - 1}dx\\
        =& B^{-1}(a, b)\int_0^1 x^{(a-k) - 1}(1 - x)^{b - 1} dx\\
        =& \frac{B(a - k, b)}{B(a, b)}  \quad \text{if $a > k$}\\
        =& \frac{\Gamma(a + b)\Gamma(a - k)}{\Gamma(a)\Gamma(a + b - k)}.
    \end{align*}
\end{proof}


\section{Scree plot for real data network}
In Figure \ref{fig: Eigen Away}, we show the scree plots, for the adjacency matrices of the Away group as mentioned in Section \ref{sec: the away group}. We can see that eigenvalues with absolute values greater than $20$ are all positive and there are about $10$ of them. 
\begin{figure}[h]
    \centering
    \includegraphics[width=0.8\textwidth]{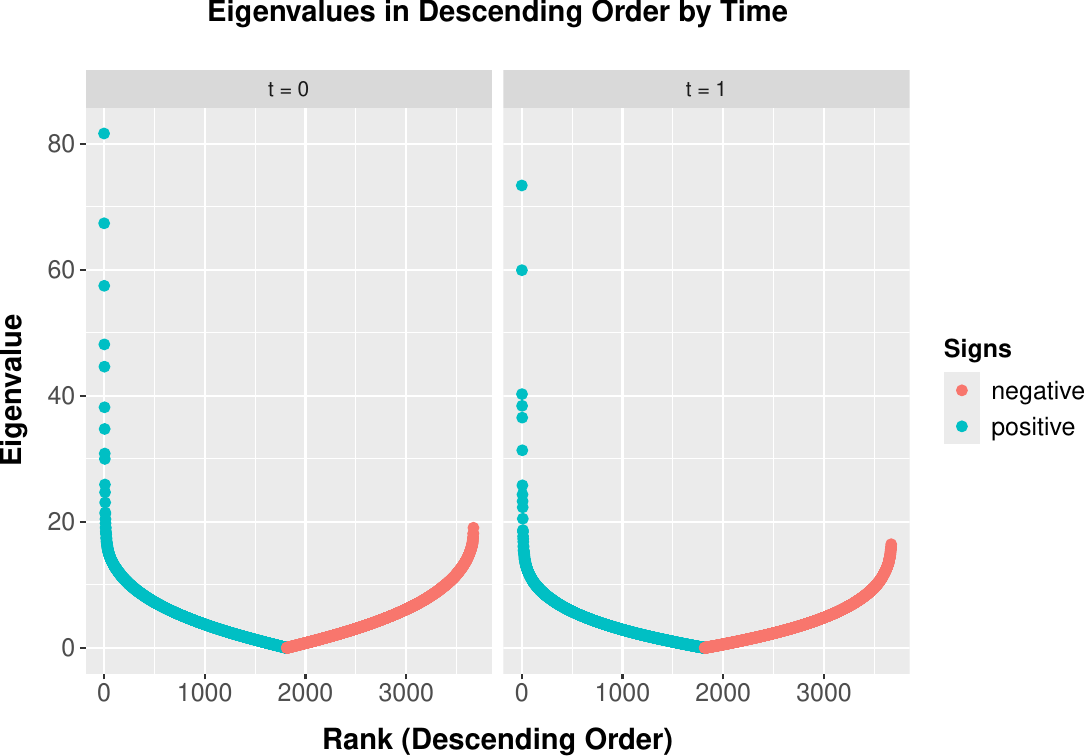}
    \caption{This is the plot of Eigenvalues vs. rank for the Away graph at period $0$ and $1$. We used this to determine the dimension to embed the adjacency matrices. We can note that eigenvalues with absolute values higer than 20 are all positive. That corresponds to top $10$-ish eigenvalues.}
    \label{fig: Eigen Away}
\end{figure}

\section{GLM Theory} \label{apd: GLM}
This section discusses the theory of generalized linear models (GLMs), including the special case of Dirichlet GLMs. For a comprehensive treatment of GLMs, see \cite{mccullagh1989generalized}. As for conditions and proofs for the consistency and asymptotic normality of GLMs, see \cite{GLM_Theory}.

\subsection{GLM Background}
Let $Y$ be a $p$-dimensional random variable in the exponential family\cite{mccullagh1989generalized} with natural parameter $\theta \in \R^p$. Then $Y$ has the following density function with respect to a $\sigma-$finite measure $\nu$:
\begin{align*}
    f(y|\theta) = \exp\left\{\theta^Tt(y) - b(\theta) + c(y)\right\},
\end{align*}
\noindent where $t(Y)$ is a sufficient statistic of $Y$\cite{casella2002statistical}.
\subsubsection{GLM Definitions}
A GLM is characterized by the following conditions\cite{GLM_Theory}:
\begin{enumerate}
    \item The response variables, $\fc{y_i}_{i = 1}^n$ are independent random variables within the same exponential family but have different natural parameters $\fc{\theta_i}_{i = 1}^n$,
    \item Explanatory variables $Z_i \in \R^{p}$ influences $y_i$  in form of a linear combination, $\gamma_i = Z_i^T \beta$, where $\beta$ is the parameter of the GLM with appropriate dimensions,
    \item $\gamma_i$ is related to $\mu(\theta_i) = \E\fb{t(y_i)}$ by some injective link function $g$, more specifically, $\gamma_i = (g \circ \mu)(\theta_i)$.
\end{enumerate}

\subsubsection{Conditions for Consistency and Asymptotic Normality} \label{subsec: GLM consistency}
In this section, we shall assume $\beta_0$ to be the true parameter. For notational convenience, the $\beta_0$ argument in any function will be omitted, e.g. $s_n(\beta_0) = s_n$. The log-likelihood of a sample $\fc{y_i}_{i = 1}^{n}$ is given by:
\begin{align*}
    \ell_n(\beta) = \s{i = 1}{n}\fr{\theta_i^T t(y_i) - b(\theta_i)} - C,\quad \theta_i = u\fr{Z_i^T \beta} \text{ for $i = 1, ..., n$}.
\end{align*}
The score function ($s_n(\beta)$), Fisher information ($F_n(\beta)$), and Hessian($H_n(\beta)$) are defined below:
\begin{align} \label{eq:GLM_computation}
    s_n(\beta) = \pderiv{\beta}{\ell_n(\beta)}^T ,\ 
    F_n(\beta) = \Var_\beta(s_n(\beta)),\ 
    H_n(\beta) = -\pderiv[2]{\beta \partial \beta^T}{\ell_n(\beta)}.
\end{align}
In addition, define:
\begin{align*}
    N_n(\delta) = \set{\beta \in \R^p}{\norm{F_n^{T/2}(\beta - \beta_0)} \leq \delta}, \text{ for $n \in \N$}.
\end{align*}
To establish consistency and asymptotic normality, we first define the following conditions\cite{GLM_Theory}:
\begin{itemize}
    \item[(D)] Divergence: $\lambda_{min}\fc{F_n} \to \infty$,
    \item[(N)] Convergence and Continuity: $\forall \delta > 0,\ \max_{\beta \in N_n(\delta)} \norm{V_n(\beta) - I} \to 0$, where $V_n(\beta) = F_n^{-1/2}H_n(\beta)F_n^{-T/2}$,
    \item[($S_\delta$)] Boundedness of the eigenvalue ratio: $\exists \text{ neighborhood } N \subset B$ of $\beta_0$ s.t. 
    \begin{align*}
        \lambda_{\min}\fc{H_n(\beta)} \geq c(\lambda_{\max}\fc{F_n})^{1/2 + \delta}, \text{ with $\beta \in N, c, \delta > 0$, and $n$ sufficiently large},
    \end{align*}
\end{itemize}

When (D) (N), $(S_{1/2})$ are all satisfied, then there exists a sequence of random variables, $\fc{\widehat{\beta}_i}_{i = 1}^{n}$ with the following properties:
\begin{itemize}
    \item[(AE)] Asymptotic Existence: $P\fr{s_n(\widehat{\beta_n}) = 0 \quad \forall n \geq n_2} = 1$,
    \item[(CP)] Consistency: $\widehat{\beta_n} \overset{a.s.}{\longrightarrow} \beta_0$,
    \item[(AN)] Asymptotic Normality: $F_n^{T/2}(\widehat{\beta_n} - \beta_0) \overset{d}{\to} N(0, I)$.
\end{itemize}
In other word, MLE asymptotically exist, it is consistent and asymptotically normal.

\subsection{Dirichlet GLM}\label{apd: Dir_GLM}
\subsubsection{The Dirichlet Distribution}
Let $\alpha\in \R^{p},\ p \geq 2$, then a random variable $X \sim \Dir(\alpha)$ ($X$ is of the Drichlet distribution with concentration parameter $\alpha$) if its probability density function is given by \cite{dirichlet_definition}:
\begin{align*}
    f_X(x) &= \frac{ \Gamma\fr{\sum\limits_{i = 1}^{p}\alpha_i} }{\prod\limits_{i = 1}^{p} \Gamma\fr{\alpha_i}} \prod_{i = 1}^{p}x_i^{\alpha_i - 1}\\
    &= \exp\fc{ \log\fr{\Gamma\fr{\mathbf{1}_p^T\alpha}} + (\alpha^T - \mathbf{1}_p^T)\log(x) - \mathbf{1}_p^T\log\fr{\Gamma(\alpha_i)} }\\
    &= \exp\fc{\alpha^T \log(x) - \fb{\mathbf{1}_p^T\log\fr{\Gamma(\alpha)}- \log\fr{\Gamma\fr{\mathbf{1}_p^T\alpha}}} - \mathbf{1}_p^T\log(x)}.
\end{align*}
\noindent where $x = (x_1, ..., x_p)$ belongs to $\Delta^{p} = \set{x\in [0,1]^p}{\mathbf{1}_p^Tx = 1}$. From the computation above, we can see that the Dirichlet distribution is in the exponential family with the natural parameter $\alpha$, and 
\begin{align*}
    b(\alpha) = \mathbf{1}_p^T\log\fr{\Gamma(\alpha)}- \log\fr{\Gamma\fr{\mathbf{1}_p^T\alpha}}.
\end{align*}
\noindent Define $\psi,\ \psi^{(1)}$ to be the digamma and trigamma function (first and second derivative of the log-Gamma function), then the mean and variance of $\log(X)$ is given by:
\begin{align*}
    \mu(\alpha) &= \E(\log(X)) = \psi(\alpha) - \psi\fr{\mathbf{1}_p\alpha},\\
    \Sigma(\alpha) &= \Var(\log(X)) = \diag\fr{\psi^{(1)}(\alpha)} - \psi^{(1)}\fr{\mathbf{1}_p\alpha}.
\end{align*}

\subsubsection[Dirichlet GLM]{A Dirichlet GLM, with link $g = \log \circ (\mu^{-1})$} \label{sec: dirGLM formula}
We shall compute everything listed in Equation \ref{eq:GLM_computation}. Let $\alpha_i \in \R^p,\ y_i \sim \Dir(\alpha_i)$ for $i = 1,..., n$. Consider a Dirichlet GLM with link $g(x) = \log\fr{\mu^{-1}(x)}$, then we have:
\begin{align*}
    \alpha_i = (g \circ \mu)^{-1}\fr{Z_i^T \beta}= \exp\fc{Z_i^T \beta},\ i = 1,..., n\quad \text{for $\fc{Z_i \in \R^{q}}_{i = 1}^n$ and $\beta \in \R^{q \times p}$}.
\end{align*}
\noindent The log-likelihood of $\beta$ are given by:
\begin{align*}
\ell (\beta|y_1, \dots, y_n)
    &= \sum^n_{i = 1} \left[ \alpha_i^{T}\log(y_i) 
    - \left[\mathbf{1}_p^T\log(\Gamma(\alpha_i)) -\log\left( \Gamma\left(\mathbf{1}_p^T\alpha_i\right)\right)  \right]
    - \mathbf{1}_p^T\log(y_i) \right].
\end{align*}

\section{Riemannian Gradient Descent on the Orthogonal Group} \label{apd: RGD}

\noindent Below, we outline how the Riemannian Gradient Descent is implemented on the orthogonal group \cite{Op_RGD} for the problem $\argmin_{W \in O_p} L\fr{W}$. It works similarly to Euclidean gradient descent, except each gradient step is taken in the tangent space using the Riemannian gradient. Then to stay in $O_p$, the result after the gradient step is retracted back to $O_p$ using a special function. For a more detailed treatment of the theory relating to optimization on smooth manifold, see \cite{boumal2023intromanifolds}.
\begin{enumerate}
    \item Initialize at some $W \in O_p$.
    \item Compute the Euclidean gradient at $W$, $L^e(W) = \pderiv{W}{L(W)}$.
    \item Compute the Riemannian gradient at $W$ that is given by the orthogonal projection of $L^e(W)$ to the tangent space of $O_p$ at $W$, $\mathcal{T}_{W}O_p$:
    \begin{enumerate}
        \item $\mathcal{T}_{W}O_p = \set{WA}{A \in \R^{p \times p} \text{ and }A^T = -A}$,
        \item The orthogonal projection is given by
        \begin{align}
            P_{\mathcal{T}_{W}O_p}(M) = W \fr{\frac{W^TM - M^TW}{2}},
        \end{align}
        \item The Riemannian gradient at $W$: $L^r(W) = P_{\mathcal{T}_{W}O_p}(L^e(W))$.
    \end{enumerate}
    \item Take a gradient descent step in the tangent space using the Riemannian gradient:
    \begin{align*}
        W^{\text{tangent}}_{t+1} = W_t - \alpha L^r(W_t), \text{ where $\alpha$ is some appropriate step size}.
    \end{align*}
    \item Retract the result from previous step back to $O_p$. This retraction is done through the matrix exponential function, $\text{Expm}$:
    \begin{align*}
        W_{t+1} = W_t\text{Expm}\fr{W_t^T W^{\text{tangent}}_{t+1} }.
    \end{align*}
    \item Iterate step $2$ to step $5$ until convergence.
\end{enumerate}

\end{document}